\documentclass[11pt]{article}

\usepackage{algorithm, algpseudocode}
\usepackage[utf8]{inputenc}

\usepackage{graphicx}
\usepackage{multirow}
\usepackage{comment}
\usepackage{amsmath,amssymb}
\usepackage{amsthm}
\usepackage{subfigure}
\usepackage[T1]{fontenc}
\usepackage{authblk}

\newtheorem{theorem}{Theorem}
\newtheorem{definition}[theorem]{Definition}
\newtheorem{lemma}[theorem]{Lemma}
\newtheorem{corollary}[theorem]{Corollary}
\theoremstyle{remark}

\newcommand*{\nred}[1]{{\#red(#1)}}
\newcommand*{\nblue}[1]{{\#blue(#1)}}
\newcommand*{\ninir}[1]{{\#ini_r(#1)}}
\newcommand*{\ninib}[1]{{\#ini_b(#1)}}
\newcommand*{\nini}[1]{{\#ini(#1)}}

\title{Uniform Bipartition in the Population Protocol Model with Arbitrary Communication Graphs} 

\author[1]{Hiroto Yasumi}
\author[1]{Fukuhito Ooshita}
\author[1]{Michiko Inoue}
\author[2]{S\'ebastien Tixeuil}

\affil[1]{Nara Institute of Science and Technology}
\affil[2]{Sorbonne Universit\'e}

\date{}

\begin{document}

\maketitle

\begin{abstract}
 In this paper, we focus on the uniform bipartition problem in the population protocol model.
This problem aims to divide a population into two groups of equal size.
In particular, we consider the problem in the context of \emph{arbitrary} communication graphs.
As a result, we clarify the solvability of the uniform bipartition problem with arbitrary communication graphs when agents in the population have designated initial states, under various assumptions such as the existence of a base station, symmetry of the protocol, and fairness of the execution.
When the problem is solvable, we present protocols for uniform bipartition. When global fairness is assumed, the space complexity of our solutions is tight. 
\end{abstract}

\section{Introduction}
\subsection{Background}

In this paper, we consider the population protocol model introduced by Angluin et al.~\cite{angluin2006computation}.
The population protocol model is an abstract model for low-performance devices.
In the population protocol model, devices are represented as anonymous agents, and a population is represented as a set of agents.
Those agents move passively (\emph{i.e.}, they cannot control their movements), and when two agents approach, they are able to communicate and update their states (this pairwise communication is called an interaction). 
A computation then consists of an infinite sequence of interactions. 

Application domains for population protocols include sensor networks used to monitor live animals (each sensor is attached to a small animal and monitors \emph{e.g.} its body temperature) that move unpredictably (hence, each sensor must handle passive mobility patterns).
Another application domain is that of molecular robot networks~\cite{murata2013molecular}. 
In such systems, a large number of molecular robots collectively work inside a human body to achieve goals such as transport of medicine. 
Since those robots are tiny, their movement is uncontrollable, and robots may only maintain extremely small memory.

In the population protocol model, many researchers have studied various fundamental problems such as leader election protocols~\cite{angluin2005stably} (A population protocol solves leader election if starting from an initially uniform population of agents, eventually a single agent outputs \emph{leader}, while all others output \emph{non-leader}), counting~\cite{aspnes2017time,beauquier2015space,beauquier2007self} (The counting problem consists in counting how many agents participate to the protocol; As the agents' memory is typically constant, this number is output by a special agent that may maintain logarithmic size memory, the base station), majority \cite{angluin2008simple} (The majority problem aims to decide which, if any, initial state in a population is a majority), $k$-partition~\cite{yasumi2019population, bipartition, yasumi2019space} (The $k$-partition problem consists in dividing a population into $k$ groups of equal size), etc. 

In this paper, we focus on the uniform bipartition problem~\cite{yasumi2019uniform, bipartition, yasumi2019space}, whose goal is to divide a population into two stable groups of equal size (the difference is one if the population size is odd).
To guarantee the stability of the group, each agent eventually belongs to a single group and never changes the group after that.
Applications of the uniform bipartition include saving batteries in a sensor network by switching on only one group, or executing two tasks simultaneously by assigning one task to each group.
Contrary to previous work that considered \emph{complete} communication graphs~\cite{yasumi2019uniform, bipartition, yasumi2019space}, we consider the uniform bipartition problem over \emph{arbitrary} graphs.
In the population protocol model, most existing works consider the complete communication graph model (every pairwise interaction is feasible). However, realistic networks command studying incomplete communications graphs (where only a subset of pairwise interactions remains feasible) as low-performance devices and unpredictable movements may not yield a complete set of interactions. 
Moreover, in this paper, we assume the designated initial states (i.e., all agents share the same given initial state), and consider the problem under various assumptions such as the existence of a base station, symmetry of the protocol, and fairness of the execution.
Although protocols with arbitrary initial states tolerate a transient fault, protocols with designated initial states can usually be designed using fewer states, and exhibit faster convergence times. Actually, it was shown in \cite{bipartition} that, with arbitrary initial states, constant-space protocols cannot be constructed in most cases even assuming complete graphs.

\subsection{Related Works}
The population protocol model was proposed by Angluin et al.~\cite{angluin2006computation}, who were recently awarded the 2020 Edsger W. Dijkstra prize in Distributed Computing for their work. While the core of the initial study was dedicated to the computability of the model, subsequent works considered various problems (\emph{e.g.}, leader election, counting, majority, uniform $k$-partition) under different assumptions (\emph{e.g.}, existence of a base station, fairness, symmetry of protocols, and initial states of agents).

The \emph{leader election} problem was studied from the perspective of time and space efficiency. 
Doty and Soloveichik~\cite{doty2018stable} proved that $\Omega(n)$ expected parallel time is required to solve leader election with probability 1 if agents have a constant number of states. 
Relaxing the number of states to a polylogarithmic value, Alistarh and Gelashvili~\cite{alistarh2015polylogarithmic} proposed a leader election protocol in polylogarithmic expected stabilization time.
Then, G{\k{a}}sieniec et al.~\cite{gkasieniec2019almost} designed a protocol with $O(\log \log n)$ states and $O(\log n \cdot \log \log n)$ expected time. Furthermore, the protocol of G{\k{a}}sieniec et al.~\cite{gkasieniec2019almost} is space-optimal for solving the problem in polylogarithmic time. 
In \cite{sudo2020time}, Sudo et al. presented a leader election protocol with $O(\log n)$ states and $O(\log n)$ expected time. This protocol is time-optimal for solving the problem. 
Finally, Berenbrink et al.~\cite{berenbrink2020optimal} proposed a time and space optimal protocol that solves the leader election problem with $O(\log \log n)$ states and $O(\log n)$ expected time.
In the case of arbitrary communication graphs, it turns out that self-stabilizing leader election is impossible~\cite{angluin2008self} (a protocol is self-stabilizing if its correctness does not depend on its initial global state). This impossibility can be avoided if oracles are available~\cite{beauquier2013self,canepa2010self} or if the self-stabilization requirement is relaxed: Sudo et al.~\cite{sudo2014loosely} proposed a loosely stabilizing protocol for leader election (loose stabilization relates to the fact that correctness is only guaranteed for a very long expected amount of time).

The \emph{counting} problem was introduced by Beauquier et al.~\cite{beauquier2007self} and popularized the concept of a base station. 
Space complexity was further reduced by follow-up works~\cite{beauquier2015space,izumi2014space2}, until Aspnes et al.~\cite{aspnes2017time} finally proposed a time and space optimal protocol.
On the other hand, by allowing the initialization of agents, the counting protocols without the base station were proposed for both exact counting~\cite{berenbrink2019counting} and approximate counting~\cite{alistarh2017time,berenbrink2019counting}. 
In~\cite{alistarh2017time}, Alistarh et al. proposed a protocol that computes an integer $k$ such that $\frac{1}{2} \log n < k < 9 \log n$ in $O(\log n)$ time with high probability using $O(\log n)$ states.
After that, Berenbrink et al.~\cite{berenbrink2019counting} designed a protocol that outputs either $\lfloor \log n \rfloor$ or $\lceil \log n \rceil$ in $O(\log^2 n)$ time with high probability using  $O(\log n \cdot \log \log n)$ states. Moreover, in~\cite{berenbrink2019counting}, they proposed the exact counting protocol that computes $n$ in $O(\log n)$ time using $\tilde{O}(n)$ states with high probability.

The \emph{majority} problem was addressed under different assumptions (\emph{e.g.}, with or without failures~\cite{angluin2008simple}, deterministic~\cite{gasieniec2017deterministic,kosowski2018brief} or probabilistic~\cite{alistarh2018space,ben2020log3,berenbrink2018population,kosowski2018brief} solutions, with arbitrary communication graphs~\cite{mertzios2014determining}, etc.). 
Those works also consider minimizing the time and space complexity.
Recently, Berenbrink et al.~\cite{berenbrink2020time} shows trade-offs between time and space for the problem.

To our knowledge, the \emph{uniform $k$-partition} problem and its variants have only been considered in complete communication graphs. 
Lamani et al.~\cite{lamani2016realization} studied a group decomposition problem that aims to divide a population into groups of designated sizes. 
Yasumi et al.~\cite{yasumi2019population} proposed a uniform $k$-partition protocol 
with no base station. 
Umino et al.~\cite{tomoki2018differentiation} extended the result to the $R$-generalized partition problem that aims at dividing a population into $k$ groups whose sizes follow a given ratio $R$.
Also, Delporte-Gallet et al.~\cite{delporte2006birds} proposed a $k$-partition protocol with relaxed uniformity constraints: the population is divided into $k$ groups such that in any group, at least $n/(2k)$ agents exist, where $n$ is the number of agents.

\begin{table}[t!]
\begin{minipage}{1.0\hsize}
\caption{The minimum number of states to solve the uniform bipartition problem with designated initial states over \emph{complete} graphs~\cite{yasumi2019uniform, bipartition}.}
\begin{center}
\label{tab:desgraglobal}
\scalebox{0.8}{
\begin{tabular}{|c|c|c|c|c|}
\hline
base station & fairness & symmetry & upper bound & lower bound \\
\hline
\multirow{4}{*}{initialized/non-initialized base station} & \multirow{2}{*}{global} & asymmetric & 3 & 3 \\ 
\cline{3-5}
 &  & symmetric & 3 & 3 \\ 
\cline{2-5} 
 & \multirow{2}{*}{weak} & asymmetric & 3 & 3 \\ 
\cline{3-5}
 &  & symmetric & 3 & 3 \\ 
\cline{1-5}
\multirow{4}{*}{no base station} & \multirow{2}{*}{global} & asymmetric & 3 & 3 \\ 
\cline{3-5}
 &  & symmetric & 4 & 4 \\ 
\cline{2-5}
 & \multirow{2}{*}{weak} & asymmetric & 3 & 3 \\ 
\cline{3-5}
 &  & symmetric & \multicolumn{2}{c|}{unsolvable} \\ 
\hline
\end{tabular}
}
\end{center}
\end{minipage}

\end{table}

Most related to our work is the uniform bipartition solution for complete communication graphs provided by Yasumi et al.~\cite{yasumi2019uniform, bipartition}. 
For the uniform bipartition problem over complete graphs with designated initial states, Yasumi et al.~\cite{yasumi2019uniform, bipartition} studied space complexity under various assumptions such as: \emph{(i)} an initialized base station, a non-initialized base station, or no base station (an initialized base station has a designated initial state, while a non-initialized has an arbitrary initial state), 
\emph{(ii)} asymmetric or symmetric protocols (asymmetric protocols allow interactions between two agents with the same state to map to two resulting different states, while symmetric protocols do not allow such a behavior), and \emph{(iii)} global or weak fairness 
(weak fairness guarantees that every individual pairwise interaction occurs infinitely often, while global fairness guarantees that every recurrently reachable configuration is eventually reached).
Furthermore, they also study the solvability of the uniform bipartition problem with arbitrary initial states.
Table~\ref{tab:desgraglobal} shows the upper and lower bounds of the number of states to solve the uniform bipartition with designated initial states over complete communication graphs.

\begin{table}[t!]
\begin{minipage}{1.0\hsize}
\caption{The minimum number of states to solve the uniform bipartition problem with designated initial states over \emph{arbitrary} graphs. $P$ is a known upper bound of the number of agents, and $l \ge 3$ and $h$ are a positive integer.}
\begin{center}
\label{tab:arbgraglobal}
\scalebox{0.8}{
\begin{tabular}{|c|c|c|c|c|}
\hline
base station & fairness & symmetry & upper bound & lower bound  \\
\hline
\multirow{3}{*}{initialized/non-initialized} & \multirow{2}{*}{global} & asymmetric & 3{\bf *} & 3\textbf{$^\dag$}  \\
\cline{3-5}
\multirow{3}{*}{base station} &  & symmetric & 3{\bf *} & 3\textbf{$^\dag$}   \\
\cline{2-5} 
 & \multirow{2}{*}{weak} & asymmetric & \begin{tabular}{c}$3P+1${\bf *} \\ $3l+1$ for no $l \cdot h$ cycle {\bf *}\end{tabular}& 3\textbf{$^\dag$} \\
\cline{3-5}
 &  & symmetric & \begin{tabular}{c}$3P+1${\bf *} \\ $3l+1$ for no $l \cdot h$ cycle {\bf *}\end{tabular} & 3\textbf{$^\dag$} \\
\cline{1-5}
\multirow{4}{*}{no base station} & \multirow{2}{*}{global} & asymmetric & 4{\bf *} & 4{\bf *}  \\
\cline{3-5}
 &  & symmetric & 5{\bf *} & 5{\bf *}   \\
\cline{2-5}
 & \multirow{2}{*}{weak} & asymmetric & \multicolumn{2}{c|}{unsolvable{\bf *}}   \\
\cline{3-5}
 &  & symmetric & \multicolumn{2}{c|}{unsolvable\textbf{$^\dag$}} \\
\hline
\multicolumn{5}{r}{\textbf{*} Contributions of this paper}\\ 
\multicolumn{5}{r}{\textbf{$^\dag$} Deduced from Yasumi et al.~\cite{bipartition}}\\ 
\end{tabular}
}
\end{center}
\end{minipage}
\end{table}

There exist some protocol transformers that transform protocols for some assumptions into ones for other assumptions.
In~\cite{angluin2006computation}, Angluin et al. proposed a transformer that transforms a protocol with complete communication graphs into a protocol with arbitrary communication graphs. 
This transformer requires the quadruple state space and works under global fairness.
In this transformer, agents exchange their states even after convergence.
For the uniform bipartition problem, since agents must keep their groups after convergence, they cannot exchange their states among different groups and thus the transformer proposed in~\cite{angluin2006computation} cannot directly apply to the uniform bipartition problem. 
Bournez et al.~\cite{bournez2013population} proposed a transformer that transforms an asymmetric protocol into symmetric protocol by assuming additional states. In~\cite{bournez2013population}, only protocols with complete communication graphs were considered and the transformer works under global fairness.
We use the same idea to construct a symmetric uniform bipartition protocol under global fairness without a base station.

\subsection{Our Contributions}
In this paper, we study the solvability of the uniform bipartition problem with designated initial states over arbitrary graphs. A summary of our results is presented in Table \ref{tab:arbgraglobal}. 
Let us first observe that, as complete communication graphs are a special case of arbitrary communication graphs, the impossibility results by Yasumi et al.~\cite{bipartition} remain valid in our setting.  
With a base station (be it initialized or non-initialized) under global fairness, we extend the three states protocol by Yasumi et al.~\cite{bipartition} from complete communication graphs to arbitrary communication graphs. 
With a non-initialized base station under weak fairness, we propose a new symmetric protocol with $3P+1$ states, where $P$ is a known upper bound of the number of agents. These results yield identical upper bounds for the easier cases of asymmetric protocols and/or initialized base station.
In addition, we also show a condition of communication graphs in which the number of states in the protocol can be reduced from $3P+1$ to constant.
Concretely, we show that the number of states in the protocol can be reduced to $3l+1$ if we assume communication graphs such that every cycle either includes the base station or its length is not a multiple of $l$, where $l$ is a positive integer greater than three.
On the other hand, with no base station under global fairness, we prove that four and five states are necessary and sufficient to solve uniform bipartition with asymmetric and symmetric protocols, respectively.
In the same setting, in complete graphs, three and four states were necessary and sufficient. 
So, one additional state enables problem solvability in arbitrary communications graphs in this setting. 
With no base station under weak fairness, we prove that the problem cannot be solved, using a similar argument as in the impossibility result for leader election by Fischer and Jiang~\cite{fischer2006self}.
Overall, we show the solvability of uniform bipartition in a variety of settings for a population of agents with designated initial states assuming arbitrary communication graphs. In cases where the problem remains feasible, we provide upper and lower bounds with respect to the number of states each agent maintains, and in all cases where global fairness can be assumed, our bounds are tight.

\section{Definitions}
 \subsection{Population Protocol Model}

A population whose communication graph is arbitrary is represented by an undirected connected graph $G=(V,E)$, where $V$ is a set of agents, and $E \subseteq V \times V$ is a set of edges that represent the possibility of an interaction between two agents.
That is, two agents $u \in V$ and $v \in V$ can interact only if $(u,v) \in E$ holds.
A protocol ${\cal P}=(Q, \delta)$ consists of $Q$ and $\delta$, where $Q$ is a set of possible states for agents, and $\delta$ is a set of transitions from $Q\times Q$ to $Q\times Q$.
Each transition in $\delta$ is denoted by $(p, q) \rightarrow (p', q')$,
which means that, when an interaction between an agent $x$ in state $ p $ and an agent $y$ in state $ q $ occurs, their states become $ p'$ and $ q' $, respectively. Moreover, we say $x$ is an initiator and $y$ is a responder. When $x$ and $y$ interact as an initiator and a responder, respectively, we simply say that $x$ interacts with $y$.
Transition $(p,q) \rightarrow (p',q')$ is null if both $p=p'$ and $q=q'$ hold. We omit null transitions in the  descriptions of protocols.
Protocol ${\cal P}=(Q, \delta)$ is symmetric if, for every transition $(p, q) \rightarrow (p',q')$ in $\delta$, $(q, p) \rightarrow (q',p')$ exists in $\delta$. In particular, if a protocol ${\cal P}=(Q,\delta)$ is symmetric and transition $(p, p) \rightarrow (p',q')$ exists in $\delta$, $p'=q'$ holds. If a protocol is not symmetric, the protocol is asymmetric.
Protocol ${\cal P}=(Q,\delta)$ is deterministic if, for any pair of states $(p,q)\in Q\times Q$, exactly one transition $(p,q)\rightarrow (p', q') $ exists in $\delta$. We consider only deterministic protocols in this paper. 
 A global state of a population is called a configuration, defined as a vector of (local) states of all agents.
 A state of agent $a$ in configuration $C$, is denoted by $s(a, C)$. 
Moreover, when $C$ is clear from the context, we simply use $s(a)$ to denote the state of agent $a$.
A transition between two configurations $C$ and $C'$ is described as $C \rightarrow C'$, and means that configuration $C'$ is obtained from $C$ by a single interaction between two agents. 
 For two configurations $C$ and $C'$, if there exists a sequence of configurations $C = C_0, C_1, \ldots , C_m = C'$ such that $C_i \rightarrow C_{i+1}$ holds for every $i$ ($0 \le i < m$), we say 
 $C'$ is reachable from $C$, denoted by $C \xrightarrow{*} C'$.
 An infinite sequence of configurations $\Xi =C_0, C_1, C_2, \ldots$ is an execution of a protocol if $C_i \rightarrow C_{i+1}$ holds for every $i$ ($i \ge 0$).
An execution $\Xi$ is weakly-fair if, for each pair of agents $(v, v') \in E$, $v$ (resp. $v'$) interacts with $v'$ (resp., $v$) infinitely often\footnote{We use this definition for the lower bound under weak fairness, but for the upper bound we use a slightly weaker version. We show that our proposed protocols for weak fairness works if, for each pair of agents $(v, v') \in E$, $v$ and $v'$ interact infinitely often (i.e., for interactions by some pair of agents $v$ and $v'$, it is possible that $v$ only becomes an initiator and $v'$ never becomes an initiator).}. 
An execution $\Xi$ is globally fair if, for every pair of configurations $C$ and $C'$ such that $C \rightarrow C'$, $C'$ occurs infinitely often when $C$ occurs infinitely often. 
Intuitively, global fairness guarantees that, if configuration $C$ occurs infinitely often, then
every possible interaction in $C$ also occurs infinitely often. Then, if $C$ occurs infinitely often, $C'$ satisfying $C\rightarrow C'$ occurs infinitely often, we can deduce that $C''$ satisfying $C'\rightarrow C''$ also occurs infinitely often. Overall, with global fairness, if a configuration $C$ occurs infinitely often, then every configuration $C^*$ reachable from $C$ also occurs infinitely often.

In this paper, we consider three possibilities for the base station: initialized base station, non-initialized base station, and no base station.  
In the model with a base station, we assume that a single agent, called a base station, exists in $V$.
Then, $V$ can be partitioned into $V_b$, the singleton set containing the base station, and $V_p$, the set of agents except for the base station. 
The base station can be distinguished from other agents in $V$, although agents in $V_p$ cannot be distinguished.
Then, the state set $Q$ can be partitioned into a state set $ Q_b $ for the base station, and a state set $ Q_p $ for agents in $V_p$.
The base station has unlimited resources (with respect to the number of states), in contrast with other resource-limited agents (that are allowed only a limited number of states).
So, when we evaluate the space complexity of a protocol, we focus on the number of states $|Q_p|$ for agents in $V_p$ and do not consider the number of states $|Q_b|$ that are allocated to the base station.
In the sequel, we thus say a protocol uses $x$ states if $|Q_p|=x$ holds.
When we assume an initialized base station, the base station has a designated initial state.
When we assume a non-initialized base station, the base station has an arbitrary initial state (in $Q_b$), although agents in $V_p$ have the same designated initial state. When we assume no base station, there exists no base station and thus $V=V_p$ holds. 
 For simplicity, we use agents only to refer to agents in $V_p$ in the following sections. To refer to the base station, we always use the term base station (not an agent). 
 In the initial configuration, both the base station and the agents are not aware of the number of agents, yet they are given an upper bound $P$ of the number of agents. However, in protocols except for a protocol in Section \ref{sec:3P+1}, we assume that they are not given $P$. 

\subsection{Uniform Bipartition Problem} 
 Let $ f: Q_p \rightarrow \{red$, $blue \}$ be a function that maps a state of an agent to $red$ or $blue$. We define the color of an agent $a$ as $f(s(a))$. Then, we say that agent $a$ is $red$ (resp., $blue$) if $f(s(a))=red$ (resp., $f(s(a))=blue$) holds. If an agent $a$ has state $s$ such that $f(s) = red$ (resp., $f(s) = blue$), we call $a$ a $red$ agent (resp., a $blue$ agent). 
For some population $V$, the number of $red$ agents (resp., $blue$ agents) in $V$ is denoted by $\nred{V}$ (resp., $\nblue{V}$).
When $V$ is clear from the context, we simply write $\# red$ and $\# blue$.

A configuration $ C $ is stable with respect to the uniform bipartition if there exists a partition $\{ H_r$, $H_b\}$ of $V_p$ that satisfies the following conditions:
 
 \begin{enumerate}
 	\item $\left||H_r|-|H_b|\right| \leq 1$ holds, and
 	\item For every configuration $C^{\prime}$ such that $C \xrightarrow{*} C^{\prime}$, each agent in $H_r$ (resp., $H_b$) remains $red$ (resp., $blue$) in $C^{\prime}$.
 \end{enumerate}
 
  An execution $\Xi = C_0$, $C_1$, $C_2$, $\ldots$ solves the uniform bipartition problem if $\Xi$ includes a configuration $C_t$ that is stable for uniform bipartition.
  Finally, a protocol $ {\cal P} $ solves the uniform bipartition problem if every possible execution $ \Xi $ of protocol $ {\cal P} $ solves the uniform bipartition problem.

\section{Upper Bounds with Non-initialized Base Station}

In this section, we prove some upper bounds on the number of states that are required to solve the uniform bipartition problem over arbitrary graphs with designated initial states and a non-initialized base station.
More concretely, with global fairness, we propose a symmetric protocol with three states by extending the protocol by Yasumi et al.~\cite{bipartition} from a complete communication graph to an arbitrary communication graph.
In the case of weak fairness, we present a symmetric protocol with $3P+1$ states, where $P$ is a known upper bound of the number of agents.

\subsection{Upper Bound for Symmetric Protocols under Global Fairness}

The state set of agents in this protocol is $ Q_p=\{initial, red, blue \} $, and we assume that $ f(initial) = f(red) = red $ and $ f(blue) = blue $ hold.
The designated initial state of agents is $initial$. 
The idea of the protocol is as follows: the base station assigns $red$ and $blue$ to agents whose state is $initial$ alternately.
As the base station cannot meet every agent (the communication graph is arbitrary), the positions of state $initial$ are moved throughout the communication graph using transitions. Thus, if an agent with $initial$ state exists somewhere in the network, the base station has infinitely many chances to interact with a neighboring agent with $initial$ state.
This implies that the base station is able to repeatedly assign $red$ and $blue$ to neighboring agents with $initial$ state unless no agent anywhere in the network has $initial$ state. 
Since the base station assigns $red$ and $blue$ alternately, the uniform bipartition is completed after no agent has $initial$ state.

To make $red$ and $blue$ alternately, the base station has a state set $Q_b=\{b_{red},b_{blue}\}$. Using its current state, the base station decides which color to use for the next interaction with a neighboring agent with $initial$ state.
Now, to move the position of an $initial$ state in the communication graph, if an agent with $initial$ state and an agent with $red$ (or $blue$) state interact, they exchange their states.
This implies that eventually an agent adjacent to the base station has $initial$ state and then the agent and the base station interact (global fairness guarantees that such interaction eventually happens). 
Transition rules of the protocol are the following (for each transition rule $(p,q) \rightarrow (p',q')$, transition rule $(q,p) \rightarrow (q',p')$ exists, but we omit the description).

\begin{enumerate}
\item $(b_{red},initial)\rightarrow(b_{blue},red)$
\item $(b_{blue},initial)\rightarrow(b_{red},blue)$
\item $(blue,initial)\rightarrow(initial,blue)$
\item $(red,initial)\rightarrow(initial,red)$
\end{enumerate}

From these transition rules, the protocol converges when no agent has $initial$ state (indeed, no interaction is defined when no agent has $initial$ state).

We now prove the correctness of the protocol.

\begin{theorem}
In the population protocol model with a non-initialized base station, there exists a symmetric protocol with three states per agent that solves the uniform bipartition problem with designated initial states assuming global fairness in arbitrary communication graphs.
\end{theorem}
\begin{proof}
Recall that $\# red$ (resp., $\# blue$) denotes the number of agents in state $s$ such that $f(s)=red$ (resp., $f(s)=blue$) holds, thus initially $\# red = n$ and $\# blue = 0$ hold.
We show that eventually $|\# red - \# blue| \le 1 $ holds and no agent ever changes its state afterwards.
Transition rules 1 and 2 indicate that, if the base station has state $b_{red}$ (resp., $b_{blue}$), and an agent with $initial$ state interact, the agent state becomes $red$ (resp., $blue$), and the base station state becomes $b_{blue}$ (resp., $b_{red}$).
By repeating these transitions, agents are assigned $red$ and $blue$ alternately, and the number of agents in state $initial$ decreases.
This implies that, eventually $|\# red - \# blue| \le 1 $ holds by alternating transition rules 1 and 2 (transition rules 1 and 2 are never enabled simultaneously, and executing one disables it while enabling the other).
Transition rules 3 and 4 indicate that, if an agent with $initial$ state and an agent with $red$ (or $blue$) state interact, they exchange their states.
From the global fairness hypothesis and transition rules 3 and 4, if there exists an agent with $initial$ state, the base station and a neighboring agent with $initial$ state eventually interact. Thus, transition rules 1 and 2 occur repeatedly unless no agent has $initial$.
After all agents have $red$ or $blue$ state, no agent ever changes its color. Recall that agents are assigned $red$ and $blue$ alternately and thus $|\# red - \# blue| \le 1$ holds after those assignments.
Moreover, the arguments do not depend on the initial state of the base station, and the proposed protocol is symmetric.
Therefore, the theorem holds.
\end{proof}

Note that, under weak fairness, this protocol does not solve the uniform bipartition problem. This is because we can construct a weakly-fair execution of this protocol such that some agents keep $initial$ state infinitely often.
For example, we can make an agent keep $initial$ by constructing an execution in the following way.
\begin{itemize}
\item If the agent (in $initial$) interacts with an agent in $red$ or $blue$, the next interaction occurs between the same pair of agents. 
\end{itemize}

\subsection{Upper Bound for Symmetric Protocols under Weak Fairness}
\label{sec:3P+1}
\subsubsection{A protocol over arbitrary graphs}
We obtain a symmetric protocol under the assumption by using a similar idea of the transformer proposed in~\cite{bournez2013population}. 
The transformer simulates an asymmetric protocol on a symmetric protocol. To do this, the transformer requires additional states. 
Moreover, the transformer works with complete communication graphs. We show that one additional state is sufficient to transform the asymmetric uniform bipartition protocol into the symmetric protocol even if we assume arbitrary graphs.

\begin{algorithm}[t!]
	\caption{Uniform bipartition protocol with $3P+1$ states}         
	\label{alg2}
	\algblock{when}{End}
	\begin{algorithmic}[1]
		\renewcommand{\algorithmicrequire}{\textbf{Variables at the base station:}}
		\renewcommand{\algorithmicwhile}{\textbf{when}}
		\Require
		\Statex $RB \in \{r,b\}$: The state that the base station assigns next
		
		\renewcommand{\algorithmicrequire}{\textbf{Variables at an agent $x$:}}
		\Require
		\Statex $color_x \in \{ ini, r,b\}$: Color of the agent, initialized to $ini$
		\Statex $depth_x \in \{\bot, 1, 2, 3, \ldots, P\}$: Depth of agent $x$ in a tree rooted at the base station, initialized to $\bot$
		\While {an agent $x$ and the base station interact}
		\If{$color_x = ini$ and $depth_x = 1$ }
		\State $color_x \leftarrow RB$
		\State $RB \leftarrow \overline{RB}$
		\EndIf
		\If{$depth_x = \bot$ }
		\State $depth_x \leftarrow 1$
		\EndIf
		\EndWhile

		\While {two agents $x$ and $y$ interact}
		\If{$depth_y \neq \bot$ and $depth_x = \bot$ }
		\State $depth_x \leftarrow depth_y + 1$
		\ElsIf{ $depth_x \neq \bot$ and $depth_y = \bot$ }
		\State $depth_y \leftarrow depth_x + 1$
		\EndIf
		\If{$depth_x < depth_y$ and $color_y = ini$ }
		\State $color_y \leftarrow color_x$
		\State $color_x \leftarrow ini$
		\EndIf
		\If{$depth_y < depth_x$ and $color_x = ini$ }
		\State $color_x \leftarrow color_y$
		\State $color_y \leftarrow ini$
		\EndIf
		\EndWhile
		\renewcommand{\algorithmicrequire}{\textbf{Note:} If $depth_x = \bot$ holds, $color_x = ini$ holds.}
		\Require
	\end{algorithmic}
\end{algorithm}

In this protocol, every agent $x$ has variables $color_x$ and $depth_x$.
Variable $color_x$ represents the color of agent $x$.
That is, for an agent $x$, if $color_x=ini$ or $color_x=r$ holds, $f(s(x))=red$ holds.
On the other hand, if $color_x=b$ holds, $f(s(x))=blue$ holds.
The protocol is given in Algorithm~\ref{alg2}. Note that this algorithm does not care an initiator and a responder.

The basic strategy of the protocol is the following.
\begin{enumerate}
\item Create a spanning tree rooted at the base station. Concretely, agent $x$ assigns its depth in a tree rooted at the base station into variable $depth_x$. Variable $depth_x$ is initialized to $\bot$.
Variable $depth_x$ obtains the depth of $x$ in the spanning tree as follows:
If the base station and an agent $p$ with $depth_p = \bot$ interact, $depth_p$ becomes $1$. 
If an agent $q$ with $depth_q \neq \bot$ and an agent $p$ with $depth_p = \bot$ interact, $depth_p$ becomes $depth_q +1$.
By these behaviors, for any agent $x$, eventually variable $depth_x$ has a depth of $x$ in a tree rooted at the base station.
\item Using the spanning tree, carry the initial color $ini$ toward the base station and make the base station assign $r$ and $b$ to agents one by one.
Concretely, if agents $x$ and $y$ interact and both $depth_y < depth_x$ and $color_x = ini$ hold, $x$ and $y$ exchange their colors (i.e., $ini$ is carried from $x$ to $y$).
Hence, since $ini$ is always carried to a smaller $depth$, eventually an agent $z$ with $depth_z=1$ obtains $ini$. After that, the base station and the agent $z$ interact and the base station assigns $r$ or $b$ to $z$.
Note that, if the base station assigns $r$ (resp., $b$), the base station assigns $b$ (resp., $r$) next.
\end{enumerate}
Then, for any agent $v$, eventually $color_{v} \neq ini$ holds. Hence, there exist $\lceil n/2 \rceil$ $red$ (resp., $blue$) agents, and $\lfloor n/2 \rfloor$ $blue$ (resp., $red$) agents if variable $RB$ in the base station has $r$ (resp., $b$) as an initial value. 
Therefore, the protocol solves the uniform bipartition problem.

From now on, we demonstrate the correctness of the protocol.
Let $\# ini=|\{x| color_x=ini \}|$, $\# r=|\{x| color_x=r \}|$, and $\# b=|\{x| color_x=b \}|$.
First, we show that $\# r$, $\# b$, and $\# ini$ do not change except for an interaction between the base station and an agent with a specific state.

\begin{lemma}
\label{lem:color}
For any weakly-fair execution of Algorithm~\ref{alg2}, unless an interaction occurs between the base station and an agent $v$ such that $depth_v = 1$ and $color_v = ini$ hold, $\# r$, $\# b$, and $\# ini$ do not change.
\end{lemma}

\begin{proof}
For each interaction occurring in a weakly-fair execution of the protocol, we consider two cases.

First, we consider the case that the base station does not participate to the interaction.
From the protocol, when agents interact, the color of agents may only change in lines 17-18 or 21-22 of the pseudocode.
On lines 17-18 and 21-22, two agents $x$ and $y$ just exchange $color_x$ and $color_y$.
Thus, $\# r$, $\# b$, and $\# ini$ do not change in this case.

Next, we consider the case that the base station participates the interaction.
In the case, from the protocol, variable $color$ changes only if the base station interacts with an agent $v$ such that $depth_v = 1$ and $color_v = ini$ hold.
Thus, the lemma holds.
\end{proof}

We now show three basic properties of variable $depth$.

\begin{lemma}
\label{lem:nochangedepth}
For any agent $v$, when $depth_{v} \neq \bot$ holds in a configuration, $depth_{v}$ does not change afterwards.
\end{lemma}
\begin{proof}
From the protocol, variable $depth_v$ of agent $v$ changes only in line 7, 12, or 14 of the pseudocode.
On lines 7, 12, and 14, $depth_v$ changes only if $depth_v = \bot$ holds.
Thus, the lemma holds.
\end{proof}

\begin{lemma}
\label{lem:depthnobot}
For any agent $v$, $depth_{v} \neq \bot$ holds after some configuration in any weakly-fair execution of Algorithm~\ref{alg2}. 
\end{lemma}
\begin{proof}
By the weak fairness assumption, for any agent $v_0$ adjacent to the base station, eventually $v_0$ and the base station interact. By lines 6-7 of the pseudocode, $depth_{v_0} \neq \bot$ holds after the interaction with the base station occurs.
From Lemma \ref{lem:nochangedepth}, for any agent $v$, once $depth_v \neq \bot$ holds, $depth_v \neq \bot$ holds perpetually. 
Hence, if agent $v$ with $depth_{v} = \bot$ is adjacent to $v'$ with $depth_{v'} \neq \bot$, eventually agents $v$ and $v'$ interact by the weak fairness assumption, and after the interaction occurs, $depth_v \neq \bot$ holds.
This implies that, for any agent $v$, eventually $depth_{v} \neq \bot$ holds.
\end{proof}

\begin{lemma}
\label{lem:depth+1}
For any agent $v_1$, when $depth_{v_1}$ is neither $\bot$ nor one, then $v_1$ is adjacent to an agent $v_2$ such that $depth_{v_1} = depth_{v_2}+1$ holds.
When $depth_{v_1}$ is one, $v_1$ is adjacent to the base station.
\end{lemma}
\begin{proof}
We focus on interactions such that an agent $v$ assigns a value other than $\bot$ to $depth_v$.
By Lemmas \ref{lem:nochangedepth} and \ref{lem:depthnobot}, the number of such interactions for a given agent is exactly one.
We consider two cases.

First, we consider the case when $depth_{v}$ becomes one by the interaction.
From the protocol, $depth_{v}$ can become one only if the base station and $v$ interact.
Hence, $v$ is adjacent to the base station.

Next, we consider the case that $depth_{v}$ becomes neither $\bot$ nor one by the interaction.
From the protocol, if the interaction happens between $v$ and an agent $w$, $depth_{v}$ becomes $depth_{w}+1$ by the interaction.
By Lemma \ref{lem:nochangedepth}, $depth_w$ and $depth_{v}$ never change afterward.
Thus, the lemma holds.
\end{proof}

From the above lemmas, we show that eventually $\# ini = 0$ holds in any weakly-fair execution of the protocol.
\begin{lemma}
\label{lem:zerozero}
For any weakly-fair execution of Algorithm~\ref{alg2}, $\# ini = 0$ holds after finite time.
\end{lemma}
\begin{proof}
For the purpose of contradiction, let us assume that, after some configuration $C$, $\# ini > 0$ never decreases.
Let us consider a configuration $C'$ such that $depth_{v} \neq \bot$ holds for any agent $v$ in $C'$, and $C'$ appears after $C$. By Lemma \ref{lem:depthnobot}, such $C'$ exists.
In $C'$, let $Ini=\{x| color_x=ini \}$, and let $v_1 \in Ini$ be an agent such that $depth_{v_1} = \min \{ depth_x| x \in Ini \}$ holds.
By Lemma \ref{lem:depth+1}, either $v_1$ is adjacent to an agent $v_2$ such that $depth_{v_2} = depth_{v_1} - 1$ or $depth_{v_1} = 1$ holds.

First, we consider the case when $depth_{v_1} = 1$ holds.
By Lemma \ref{lem:depth+1}, $v_1$ is adjacent to the base station.
Since $depth_{v_1} = 1$ holds, there is not an agent $v'_2$ such that $depth_{v'_2} < depth_{v_1}$ holds.
Hence, from the protocol, $color_{v_1}$ keeps $ini$ unless $v_1$ and the base station interact.
Thus, eventually $v_1$ with $color_{v_1}=ini$ and the base station interact.
Then, $\# ini$ decreases in this case.

Next, we consider the case that $v_1$ is adjacent to an agent $v_2$ such that $depth_{v_2} = depth_{v_1} - 1$.
From the protocol, since $depth_{v_1} \neq 1$ holds, $color_{v_1}$ keeps $ini$ unless an interaction happens between $v_1$ and an agent $v'_2$ such that $depth_{v'_2} < depth_{v_1}$ holds.
From the weak fairness, eventually $v_1$ and such $v'_2$ interact and $color_{v'_2}$ becomes $ini$.
At that time, the smallest depth of agents with color $ini$ decreases. By repeating this behavior similarly, eventually some agent $v_h$ with $depth_{v_h}=1$ obtains color $ini$. After that, $v_h$ and the base station interact and $\# ini$ decreases.
This is a contradiction.
\end{proof}

Next, we prove that $|\#r - \# b| \le 1$ holds in any weakly-fair execution of the protocol.
\begin{lemma}
\label{lem:redblue}
For any configuration in any weakly-fair execution of Algorithm~\ref{alg2}, either \emph{(i} both $\#r - \# b = 1$ and $RB = b$ hold, or \emph{(ii)} both $\#r - \# b = 0$ and $RB = r$ hold.
\end{lemma}

\begin{proof}
Consider an execution $\Xi=C_0$, $C_1$, $\ldots$ of the protocol.
We prove the lemma by induction on the index of a configuration.
In the base case ($C_0$), $\#r = \# b = 0$ and $RB = r$ hold and thus alternative \emph{(ii)} of the lemma holds immediately.

For the induction step, assume that there exists an integer $i\geq 0$ such that the lemma holds in $C_i$.
Consider an interaction at $C_i \rightarrow C_{i+1}$.
We consider two cases.

First, we consider the case where, for an agent $v$ such that $depth_v = 1$ and $color_v = ini$ hold, the base station and $v$ do not interact.
By Lemma \ref{lem:color}, when the base station and such $v$ do not interact, $\# r$, $\# b$, and $\# ini$ do not change.
Moreover, from the protocol, $RB$ does not change at the interaction, and thus either alternatives \emph{(i)} of \emph{(ii)} of the lemma holds in this case.

Next, we consider the case where, for an agent $v$ such that $depth_v = 1$ and $color_v = ini$ hold, the base station and $v$ interact.
\begin{itemize}
\item In the case where $RB = r$ holds in $C_i$: After the interaction at $C_i \rightarrow C_{i+1}$, $\# r$ increases by one and $RB$ becomes $b$.
By the induction assumption, since $\#r - \# b = 0$ holds in $C_i$, $\#r - \# b = 1$ holds in $C_{i+1}$.
Hence, alternative \emph{i)} of the lemma holds.
\item In the case where $RB = b$ holds in $C_i$: After the interaction at $C_i \rightarrow C_{i+1}$, $\# b$ increases by one and $RB$ becomes $r$.
By the induction assumption, since $\#r - \# b = 1$ holds in $C_i$, $\#r - \# b = 0$ holds in $C_{i+1}$.
Hence, alternative \emph{(ii)} of the lemma holds.
\end{itemize}
Thus, the lemma holds.
\end{proof}

Using Lemmas \ref{lem:zerozero} and \ref{lem:redblue}, we show that the protocol solves the uniform bipartition problem.
\begin{theorem}
\label{the:pos3p}
Algorithm~\ref{alg2} solves the uniform bipartition problem.
That is, there exists a protocol with $3P+1$ states and designated initial states that solves the uniform bipartition problem under weak fairness assuming arbitrary communication graphs with a non-initialized base station.
\end{theorem}

\begin{proof}
By Lemmas \ref{lem:zerozero} and \ref{lem:redblue}, $|\#r - \# b| \le 1$ holds in any weakly-fair execution of the protocol and eventually $\# ini = 0$ holds in the execution.
Moreover, from the protocol, when there exists no agent $v$ such that $color_{v}=ini$, any agent $x$ does not change its $color_x$.
Thus, the protocol solves the problem. Additionally, the protocol works with $2P+(P+1)$ states.
This is because, if $depth_x = \bot$ holds, $color_x = ini$ holds. That is, $depth_x$ takes $\bot$, $1$, $2$, $\ldots$, $P$ for $ini$, and takes $1$, $2$, $\ldots$, $P$ for $r$ and $b$.

Therefore, the theorem holds.
\end{proof}

\subsubsection{A protocol with constant states over a restricted class of graphs}

In this subsection, we show that the space complexity of Algorithm~\ref{alg2} can be reduced to constant for communication graphs such that every cycle either includes the base station or its length is not a multiple of $l$, where $l$ is a positive integer at least three.

We modify Algorithm~\ref{alg2} as follows. Each agent maintains the distance from the base station by computing modulo $l$ plus 1. That is, we change lines 12 and 14 in Algorithm~\ref{alg2} to $depth_x \leftarrow depth_y \mod l + 1$ and $depth_y \leftarrow depth_x \mod l + 1$, respectively.  Now $depth_x \in \{\bot,1,2,3$, $\ldots$, $l\}$ holds for any agent $x$. Then we redefine the relation $depth_x<depth_y$ in lines 16 and 20 as follows: $depth_x<depth_y$ holds if and only if either $depth_x=1 \land depth_y=2$, $depth_x=2 \land depth_y=3$, $depth_x=3 \land depth_y=4$, $\ldots$, $depth_x=l-1 \land depth_y=l$, or $depth_x=l \land depth_y=1$ holds. 

We can easily observe that these modifications do not change the essence of Algorithm\ref{alg2}. For two agents $x$ and $y$, we say $x<y$ if $depth_x<depth_y$ holds. Each agent $x$ eventually assigns a depth of $x$ modulo $l$ plus 1 to $depth_x$, and at that time there exists a path $x_0,x_1,\ldots,x_h$ such that $x_0$ is a neighbor of the base station, $x=x_h$ holds, and $x_i<x_{i+1}$ holds for any $0\le i<h$. In addition, there exists no cycle $x_0,x_1,\ldots,x_h=x_0$ such that $x_i<x_{i+1}$ holds for any $0\le i<h$. This is because, from the definition of relation '$<$', the length of such a cycle should be a multiple of $l$, but we assume that underlying communication graphs do not include a cycle of agents in $V_p$ whose length is a multiple of $l$. Hence, similarly to Algorithm~\ref{alg2}, we can carry the initial color $ini$ toward the base station and make the base station assign $r$ and $b$ to agents one by one. 

\begin{corollary}
There exists a protocol with $3l+1$ states and designated initial states that solves the uniform bipartition problem under weak fairness assuming arbitrary communication graphs with a non-initialized base station if, for any cycle of the communication graphs, it either includes the base station or its length is not a multiple of $l$, where $l$ is a positive integer at least three.
\end{corollary}

\section{Upper and Lower Bounds with No Base Station}
In this section, we show upper and lower bounds of the number of states to solve the uniform bipartition problem with no base station and designated initial states over arbitrary communication graphs.
Concretely, under global fairness, we prove that the minimum number of states for asymmetric protocols is four, and the minimum number of states for symmetric protocols is five.
Under weak fairness, we prove that the uniform bipartition problem cannot be solved without a base station using proof techniques similar to those Fischer and Jiang~\cite{fischer2006self} used to show the impossibility of leader election.

\subsection{Upper Bound for Asymmetric Protocols under Global Fairness}
\label{sec:Asy}
In this subsection, over arbitrary graphs with designated initial states and no base station under global fairness, we give an asymmetric protocol with four states.

\begin{algorithm}[t!]
	\caption{Uniform bipartition protocol with four states}         
	\label{alg:arbAsy}
	\algblock{when}{End}
	\begin{algorithmic}[1]
		\renewcommand{\algorithmicrequire}{\textbf{A state set}}
		\Require
		\Statex $Q = \{ r^{\omega}, b^{\omega}, r, b \}$
		\renewcommand{\algorithmicrequire}{\textbf{A mapping to colors}}
		\Require
		\Statex $f(r^{\omega})=f(r) = red$
		\Statex $f(b^{\omega})=f(b) = blue$

		\renewcommand{\algorithmicrequire}{\textbf{Transition rules}}
		\Require
		\Statex
		\begin{enumerate}
			\item $(r^{\omega},r^{\omega})\rightarrow(r,b)$
			\item $(r^{\omega},b^{\omega})\rightarrow(b,b)$
			\item $(r^{\omega},r)\rightarrow(r,r^{\omega})$
			\item $(b^{\omega},b)\rightarrow(b,b^{\omega})$
			\item $(r^{\omega},b)\rightarrow(r,b^{\omega})$
			\item $(b^{\omega},r)\rightarrow(b,r^{\omega})$
		\end{enumerate}

	\end{algorithmic}
\end{algorithm}

We define a state set of agents as $Q = \{ r^{\omega}, b^{\omega}, r, b \}$ , and function $f$ as follows:
$f(r^{\omega})=f(r) = red$ and $f(b^{\omega})=f(b) = blue$.
We say an agent has a token if its state is $r^{\omega}$ or $b^{\omega}$.
Initially, every agent has state $r^\omega$, that is, every agent is $red$ and has a token.
The transition rules are given in Algorithm \ref{alg:arbAsy} (for each transition rule $(p,q) \rightarrow (p',q')$, transition rule $(q,p) \rightarrow (q',p')$ exists, but we omit the description except for transition rule 1).
 
The basic strategy of the protocol is as follows.
When two agents with tokens interact and one of them is $red$, a $red$ agent transitions to $blue$ and the two tokens are deleted (transition rules 1 and 2).
Since $n$ tokens exist initially and the number of tokens decreases by two in an interaction, $ \lfloor n/2 \rfloor$ $blue$ agents appear and $\lceil n/2 \rceil$ $red$ agents remain after all tokens (except one token for the case of odd $n$) disappear.
To make such interactions, the protocol moves a token when agents with and without a token interact (transition rules 3, 4, 5, and 6). Global fairness guarantees that, if two tokens exist, an interaction of transition rule 1 or 2 happens eventually.
Therefore, the uniform bipartition is achieved by the protocol.

From now, we prove the correctness of the protocol.
We define $\# r$, $\# b$, $\# r^{\omega}$, $\# b^{\omega}$ as the number of agents that have state $r$, $b$, $r^{\omega}$, $b^{\omega}$, respectively.
Let $\# red= \# r + \# r^{\omega}$ and $\# blue= \# b + \# b^{\omega}$ be the number of $red$ and $blue$ agents, respectively.
Let $\# token= \# r^{\omega} + \# b^{\omega}$ be the number of agents with tokens.

\begin{lemma}
\label{lem:rb2}
In any globally-fair execution of Algorithm~\ref{alg:arbAsy}, $\# r  = \# b + 2* \# b^{\omega} $ holds in any configuration.
\end{lemma}

\begin{proof}
Consider an execution $\Xi=C_0$, $C_1$, $\ldots$ of the protocol.
We prove the equation by induction on the index of a configuration.
The base case is the case of $C_0$.
In this case, the equation holds because all agents have $r^{\omega}$ initially.
For the induction step, assume that the equation holds in $C_i (0 \le i)$.
Consider an interaction at $C_i \rightarrow C_{i+1}$ for each transition rule.

\begin{itemize}
\item Transition rule 1: When the transition rule 1 occurs at $C_i \rightarrow C_{i+1}$, $\# r$ and $\# b$ increase by one.
\item Transition rule 2: When the transition rule 2 occurs at $C_i \rightarrow C_{i+1}$, $\# b^{\omega}$ decreases by one, and $\# b$ increases by two.
\item Transitions rule 3 and 4: When the transition rule 3 or 4 occurs at $C_i \rightarrow C_{i+1}$, $\#r$, $\#b$, and $\#b^\omega$ do not change.
\item Transition rule 5: When the transition rule 5 occurs at $C_i \rightarrow C_{i+1}$, $\# b$ decreases by one, and $\# r$ and $\# b^{\omega}$ increase by one.
\item Transition rule 6: When the transition rule 6 occurs at $C_i \rightarrow C_{i+1}$, $\# b$ increases by one, and $\# r$ and $\# b^{\omega}$ decrease by one.
\end{itemize}
For every case, $\# r  = \# b + 2* \# b^{\omega} $ holds in $C_{i+1}$.
Therefore, the lemma holds.
\end{proof}

\begin{lemma}
\label{lem:tok}
For any globally-fair execution of the Algorithm~\ref{alg:arbAsy}, $\# token \le 1$ holds after finite time. 
\end{lemma}
\begin{proof}
First of all, there is no transition rule that increases $\# token$.
This implies that, if $\# token \le 1$ holds at some configuration, $\# token \le 1$ holds thereafter. 
Hence, for the purpose of contradiction, we assume that there exists a globally-fair execution $\Xi$ of the protocol where $\#token = x > 1$ continuously holds after some configuration.

Consider a configuration $C$ that occurs infinitely often in $\Xi$.
Note that $C$ is stable and satisfies $\#token = x$.
Consider two agents $a_1$ and $a_2$ such that $a_1$ is adjacent to $a_2$ and one of them is $red$ in $C$.

Since tokens can move through a graph by swapping states (transition rules 3,4,5, and 6), a configuration $C'$ such that $a_1$ and $a_2$ have a token is reachable from $C$.
When $a_1$ and $a_2$ interact at $C' \rightarrow C''$, $\#token = x - 2 $ holds in $C''$.
From the global fairness assumption, since $C$ occurs infinitely often in $\Xi$, $C''$ also occurs infinitely often.
Since $\#token = x $ continuously holds after some configuration, this is a contradiction.
\end{proof}

Next, by using these lemmas, we show that Algorithm~\ref{alg:arbAsy} solves the problem under the assumptions.

\begin{theorem}
\label{the:pos4}
Algorithm~\ref{alg:arbAsy} solves the uniform bipartition problem.
That is, there exists a protocol with four states and designated initial states that solves the uniform bipartition problem under global fairness over arbitrary communication graphs.
\end{theorem}

\begin{proof}
Since $\# token$ is reduced only by transition rules 1 and 2, $\# token$ is reduced by two in an interaction.
This implies that, by Lemma \ref{lem:tok}, when $n$ is even (resp., odd), $\# token = 0$ (resp.,  $\# token = 1$) holds after some configuration $C$.

First, we consider the case that $n$ is even.
By Lemma \ref{lem:rb2}, if $\# token = 0$ holds, $\# r = \# b$ holds.
Hence, since $n = \# r + \# b + \# token$ holds, $\# r = \# b = n/2$ holds at $C$.
Moreover, since agents can change their colors only if $\# token \ge 2$ holds, they do not change their colors after $C$. Hence $C$ is a stable configuration, and thus the uniform bipartition is completed.

Next,  we consider the case that $n$ is odd.
When $\# token = 1$ holds, we consider two cases.
If an agent in state $r^\omega$ exists at $C$, by Lemma \ref{lem:rb2}, $\# r = \# b$ holds and thus $\# red = \# r + \# r^{\omega} = \# b + 1 = \# blue +1$ holds.
If an agent in state $b^\omega$ exists at $C$, by Lemma \ref{lem:rb2}, $\# r = \# b + 2$ holds and thus $\# red = \# r =  \# b + \# b^{\omega} +1  = \# blue +1$ holds.
Hence, in both cases, $\# red - \# blue =1$ holds at $C$.
Since agents can change their colors only if $\# token \ge 2$ holds, they do not change their colors after $C$. Hence $C$ is a stable configuration, and thus the uniform bipartition is completed.
\end{proof}

\subsection{Upper Bound for Symmetric Protocols under Global Fairness}
\label{sec:uppsym}

In this subsection, with arbitrary communication graphs with designated initial states and no base station under global fairness, we give a symmetric protocol with five states.

Observe that, with designated initial states and no base station, clearly no symmetric protocol can solve the problem if the number of agents $n$ is two (the state of the two agents is the same in the initial state, so symmetry is never broken and uniform bipartition cannot occur).
Thus, we assume that $3 \le n$ holds.

\begin{algorithm}[t!]
	\caption{Uniform bipartition protocol with five states}         
	\label{alg:arbSym}
	\algblock{when}{End}
	\begin{algorithmic}[1]
		\renewcommand{\algorithmicrequire}{\textbf{A state set}}
		\Require
		\Statex $Q = \{ r^{\omega}_0, r^{\omega}_1, b^{\omega}, r, b \}$

		\renewcommand{\algorithmicrequire}{\textbf{A mapping to colors}}
		\Require
		\Statex $f(r^{\omega}_0)= f(r^{\omega}_1) = f(r) = red$
		\Statex $f(b^{\omega})=f(b) = blue$
		\renewcommand{\algorithmicrequire}{\textbf{Transition rules}}
		\Require
		\Statex
		\begin{enumerate}
			\item $(r^{\omega}_0,r^{\omega}_0)\rightarrow(r^{\omega}_1,r^{\omega}_1)$
			\item $(r^{\omega}_1,r^{\omega}_1)\rightarrow(r^{\omega}_0,r^{\omega}_0)$
			\item $(r^{\omega}_0,r^{\omega}_1)\rightarrow(r,b)$
			\item $(r^{\omega}_0,r)\rightarrow(r,r^{\omega}_0)$
			\item $(r^{\omega}_1,r)\rightarrow(r,r^{\omega}_0)$
			\item $(b^{\omega},b)\rightarrow(b,b^{\omega})$
			\item $(r^{\omega}_0,b)\rightarrow(r,b^{\omega})$
			\item $(r^{\omega}_1,b)\rightarrow(r,b^{\omega})$
			\item $(b^{\omega},r)\rightarrow(b,r^{\omega}_0)$
			\item $(r^{\omega}_0,b^{\omega})\rightarrow(b,b)$
			\item $(r^{\omega}_1,b^{\omega})\rightarrow(b,b)$
		\end{enumerate}

	\end{algorithmic}
\end{algorithm}
We define a state set of agents as $Q = \{ r^{\omega}_0, r^{\omega}_1, b^{\omega}, r, b \}$ , and function $f$ as follows:
$f(r^{\omega}_0)= f(r^{\omega}_1) = f(r) = red$ and $f(b^{\omega})=f(b) = blue$.
We say an agent has a token if its state is $r^{\omega}_0$, $r^{\omega}_1$ , or $b^{\omega}$.
Initially, every agent has state $r^{\omega}_0$, that is, every agent is $red$ and has a token.
The transition rules are given in Algorithm \ref{alg:arbSym}.

The idea of Algorithm~\ref{alg:arbSym} is similar to Algorithm \ref{alg:arbAsy}.
That is, when two agents with tokens interact and one of them is $red$, a $red$ agent transitions to $blue$, and the two tokens are deleted.
Then, eventually $ \lfloor n/2 \rfloor$ $blue$ agents appear and $\lceil n/2 \rceil$ $red$ agents remain after all tokens (except one token for the case of odd $n$) disappear.
However,  to make a $red$ agent transition to a $blue$ agent in the first place, Algorithm \ref{alg:arbAsy} includes transition rule 1 that makes agents with the same states transition to different states. This implies that Algorithm~\ref{alg:arbAsy} is not symmetric. 
Hence, by borrowing the technique proposed in~\cite{bournez2013population}, we improve Algorithm \ref{alg:arbAsy} so that the new protocol (Algorithm~\ref{alg:arbSym}) makes a $red$ agent transition to a $blue$ agent without such a transition (and two tokens are deleted at that time).
Concretely, we realize it as follows.
In Algorithm~\ref{alg:arbSym}, there are two states $r^{\omega}_0$ and $r^{\omega}_1$ that are $red$ and have a token.
When two agents with $r^{\omega}_0$ interact, they transition to $r^{\omega}_1$, and vice versa (transition rules 1 and 2).
Thus, under global fairness, eventually an agent with $r^{\omega}_0$ interacts with an agent with $r^{\omega}_1$ and then one of them transitions to $blue$, and two tokens are deleted (transition rule 3).
Observe that, these transitions do not affect the essence of Algorithm \ref{alg:arbAsy}.
This is because the numbers of $blue$ agents, $red$ agents, and tokens do not change after transition rules 1 and 2, and a $red$ agent transitions to $blue$, and two tokens are deleted at transition rule 3.

We now prove the correctness of Algorithm~\ref{alg:arbSym} along the proof in subsection \ref{sec:Asy}.
We define $\# r$, $\# b$, $\# r^{\omega}_0$, $\# r^{\omega}_1$, $\# b^{\omega}$ as the number of agents that have state $r$, $b$, $r^{\omega}_0$, $\# r^{\omega}_1$, $b^{\omega}$, respectively.
Let $\# red= \# r + \# r^{\omega}_0 + \# r^{\omega}_1$ and $\# blue= \# b + \# b^{\omega}$ be the number of $red$ and $blue$ agents, respectively.
Let $\# token= \# r^{\omega}_0 + \# r^{\omega}_1 + \# b^{\omega}$ be the number of agents with tokens.

Recall that a mechanism of symmetry breaking (transition rules 1, 2, and 3) does not affect the essence of Algorithm \ref{alg:arbAsy}.
In particular, if $r^{\omega}_0 = r^{\omega}_1 ( = r^{\omega} )$ holds, Algorithm \ref{alg:arbSym} is equal to Algorithm \ref{alg:arbAsy}.
Thus, since an equation of Lemma \ref{lem:rb2} holds in any configuration of any execution of Algorithm \ref{alg:arbSym} and $\# r^{\omega}$ is not related to the equation, the following corollary holds.

\begin{corollary}
\label{cor:rb2-2}
In any globally-fair execution of Algorithm \ref{alg:arbSym}, $\# r  = \# b + 2* \# b^{\omega} $ holds in any configuration.
\end{corollary}

Moreover, we show the following lemma similarly to Lemma \ref{lem:tok}.

\begin{lemma}
\label{lem:tok-2}
For any globally-fair execution of Algorithm \ref{alg:arbSym}, $\# token \le 1$ holds after some configuration.
\end{lemma}
\begin{proof}
First of all, there is no transition rule that increase $\# token$.
This implies that, if $\# token \le 1$ holds at some configuration, $\# token \le 1$ holds thereafter. 
Hence, for the purpose of contradiction, we assume that there exists a globally-fair execution $\Xi$ of Algorithm \ref{alg:arbSym} where $\#token = x > 1$ continuously holds after some configuration.

Consider a configuration $C$ that occurs infinitely often in $\Xi$.
Note that $C$ is stable and satisfies $\#token = x$.
First, we show that at least one $blue$ agent occurs from the initial configuration.
When $n = \# r^{\omega}_0 + \# r^{\omega}_1 $ holds, only transition rules 1, 2, and 3 can occur.
From the global fairness assumption, an agent with $r^{\omega}_0$ is adjacent to agent with $r^{\omega}_1$ at the same time and then they interact.
By the interaction, transition rule 3 happens and thus one $blue$ agent appears.

Since there is no transition rule that decreases the number of $blue$ agents, at least one $blue$ agent exists after some configuration of $\Xi$.
Moreover, by Corollary \ref{cor:rb2-2}, if there exists $blue$ agent, there also exists $red$ agent.
Hence, in $C$, there exist at least one $blue$ agent and at least one $red$ agent.
Consider two agents $a_1$ and $a_2$ such that $a_1$ is adjacent to $a_2$ and $a_1$ (resp., $a_2$) is $red$ (resp., $blue$) in $C$.

Since tokens can move through a graph by swapping states (transition rules 4, 5, 6, 7, 8, and 9), a configuration $C'$ such that $a_1$ and $a_2$ have a token is reachable from $C$.
When $a_1$ and $a_2$ interact at $C' \rightarrow C''$, $\# token = x - 2 $ holds in $C''$.
From the global fairness assumption, since $C$ occurs infinitely often in $\Xi$, $C''$ also occurs infinitely often.
Since $\#token = x $ continuously holds after some configuration, this is a contradiction.
\end{proof}

Finally, we show, similarly to Theorem \ref{the:pos4}, that Algorithm \ref{alg:arbSym} solves the uniform bipartition problem.

\begin{theorem}
Algorithm~\ref{alg:arbSym} solves the uniform bipartition problem.
That is, there exists a symmetric protocol with five states and designated initial states that solves the uniform bipartition problem under global fairness with arbitrary communication graphs.
\end{theorem}
\begin{proof}

Since $\# token$ is reduced only in transition rules 3, 10, and 11, $\# token$ is reduced by two in an interaction.
This implies that, by Lemma \ref{lem:tok-2}, when $n$ is even (resp., odd), $\# token = 0$ (resp.,  $\# token = 1$) holds after some configuration $C$.

First, we consider the case where $n$ is even.
By Corollary \ref{cor:rb2-2}, if $\# token = 0$ holds, $\# r = \# b$ holds.
Hence, since $n = \# r + \# b + \# token$ holds, $\# r = \# b = n/2$ holds at $C$.
Moreover, since agents can change their colors only if $\# token \ge 2$ holds, they do not change their colors after $C$. Hence $C$ is a stable configuration, and thus the uniform bipartition is completed.

Next, we consider the case that $n$ is odd.
When $\# token = 1$ holds, we consider two cases.
If an agent in state $r^{\omega}_0$ or $r^{\omega}_1$ exists at $C$, by Corollary \ref{cor:rb2-2}, $\# r = \# b$ holds and thus $\# red = \# r + \# r^{\omega}_0 + \# r^{\omega}_1 = \# b + 1 = \# blue +1$ holds.
If an agent in state $b^\omega$ exists at $C$, by Corollary \ref{cor:rb2-2}, $\# r = \# b + 2$ holds and thus $\# red = \# r =  \# b + \# b^{\omega} +1  = \# blue +1$ holds.
Hence, in both cases, $\# red - \# blue =1$ holds at $C$.
Since agents can change their colors only if $\# token \ge 2$ holds, they do not change their colors after $C$. Hence $C$ is a stable configuration, and thus the uniform bipartition is completed.
\end{proof}

\subsection{Lower Bound for Asymmetric Protocols under Global Fairness}
In this section, we show that, over arbitrary graphs with designated initial states and no base station under global fairness, there exists no asymmetric protocol with three states.

To prove this, we first show that, when the number of agents $n$ is odd and no more than $P/2$, each agent changes its own state to another state infinitely often in any globally-fair execution $\Xi$ of a uniform bipartition protocol $Alg$, where $P$ is a known upper bound of the number of agents. This proposition holds regardless of the number of states in a protocol.

After that, we prove impossibility of an asymmetric protocol with three states.
The outline of the proof is as follows. 
For the purpose of contradiction, we assume that there exists a protocol $Alg$ that solves the problem with three states.
From the above proposition, in any globally-fair execution, some agents change their state infinitely often.
Now, with three states, the number of $red$ or $blue$ states is at least one and thus, 
if we assume without loss of generality that the number of $blue$ states is one, 
agents with the $blue$ state change their color eventually.
This is a contradiction.


From now on, we show that, in any globally-fair execution $\Xi$ of a protocol $Alg$ solving uniform bipartition over an arbitrary communication graph such that the number of agents $n < P/2$ is odd, all agents transition their own state to another state infinitely often.
\begin{lemma}
\label{lem:nocha}
Assume that there exists a uniform bipartition protocol $Alg$ with designated initial states over arbitrary communication graphs assuming global fairness.
Consider a graph $G=(V,E)$ such that the number of agents $n$ is odd and no more than $P/2$.
In any globally-fair execution $\Xi =C_0$, $C_1$, $\ldots$ of $Alg$ over $G$, each agent changes its state infinitely often.
\end{lemma}

\begin{proof}
The outline of the proof is as follows.
First, for the purpose of contradiction, we assume that there exists an agent $v_{\alpha}$ that never changes its state after some stable configuration $C_h$ in a globally-fair execution $\Xi$ over graph $G$.
Let $s_\alpha$ be a state that $v_\alpha$ has after $C_h$. Let $v_{\beta} \in V$ be an agent adjacent to $v_{\alpha}$ and $S_\beta$ be a set of states that $v_\beta$ has after $C_h$.
Since the number of states is finite, there exists a stable configuration $C_t$ that occurs infinitely often after $C_h$.
Next, let $G'_1=(V'_1, E_1)$ and $G'_2=(V'_2, E_2)$ be graphs that are isomorphic to $G$.
Moreover, let $v'_{\alpha} \in V'_1$ (resp., $v'_{n+\beta} \in V'_2$) be an agent that corresponds to $v_{\alpha} \in V$ (resp., $v_\beta \in V$).
We construct $G'=(V',E')$ by connecting $G'_1$ and $G'_2$ with an additional edge $(v'_{\alpha},v'_{n+\beta})$.
Over $G'$, we consider an execution $\Xi'$ such that, agents in $G'_1$ and $G'_2$ behave similarly to $\Xi$ until $C_t$ occurs in $G'_1$ and $G'_2$, and then make interactions so that $\Xi'$ satisfies global fairness.
Since $\Xi$ is globally fair, we can show the following facts after $G'_1$ and $G'_2$ reach $C_t$ in $\Xi'$.
\begin{itemize}
\item $v'_{\alpha}$ has state $s_\alpha$ as long as $v'_{n+\beta}$ has a state in $S_\beta$.
\item $v'_{n+\beta}$ has a state in $S_\beta$ as long as $v'_{\alpha}$ has state $s_\alpha$.
\end{itemize}
From these facts, in $\Xi'$, $v'_{\alpha}$ continues to have state $s_\alpha$ and $v'_{n+\beta}$ continues to have a state in $S_\beta$.
Hence, in $\Xi'$, each agent in $V'_1$ cannot notice the existence of agents in $V'_2$, and vice versa.
This implies that, in stable configurations, $\nred{V} = \nred{V'_1} = \nred{V'_2}$ and $\nblue{V} = \nblue{V'_1} = \nblue{V'_2}$ hold.
Since the number of agents in $G$ is odd, $\nred{V} - \nblue{V} = 1 $ or $\nblue{V} - \nred{V} = 1 $ holds in stable configurations of $\Xi$.
Thus, in stable configurations of $\Xi'$, $|\nred{V'} - \nblue{V'}| = 2$ holds.
Since $\Xi'$ is globally fair, this is a contradiction.

From now on, we show the details of the proof. 

Let $V=\{v_0$, $v_1$, $v_2$, $v_3$, $\ldots$, $v_{n-1}\}$.
Assume, for the purpose of contradiction, that there exists $v_{\alpha}$ that does not change its state after some stable configuration $C_h$ in globally-fair execution $\Xi$.
Let $s_\alpha$ be a state that $v_\alpha$ has after $C_h$. Since the number of states is finite, in $\Xi$, there exists a stable configuration $C_{t}$ that appears infinitely often after $C_h$.
Without loss of generality, $\nred{V} - \nblue{V} = 1 $ holds after $C_h$.
Let $v_{\beta}$ be an agent that is adjacent to $v_{\alpha}$.

Next, consider a communication graph $G'=(V', E')$ that satisfies the following.
\begin{itemize}
\item $V'=V'_1 \cup V'_2$, where $V'_1 = \{v'_0$, $v'_1$, $v'_2$, $\ldots$, $v'_{n-1}\}$ and $V'_2 = \{v'_{n}$, $v'_{n+1}$, $v'_{n+2}$, $\ldots$, $v'_{2n-1}\}$.
\item $E'=\{(v'_x,v'_y),(v'_{x+n},v'_{y+n}) \in V' \times V' \mid (v_x, v_y) \in E \} \cup \{(v'_{\alpha}, v'_{n+\beta}) \}$.
\end{itemize}

\begin{figure}[t!]
\begin{center}
\includegraphics[scale=0.4]{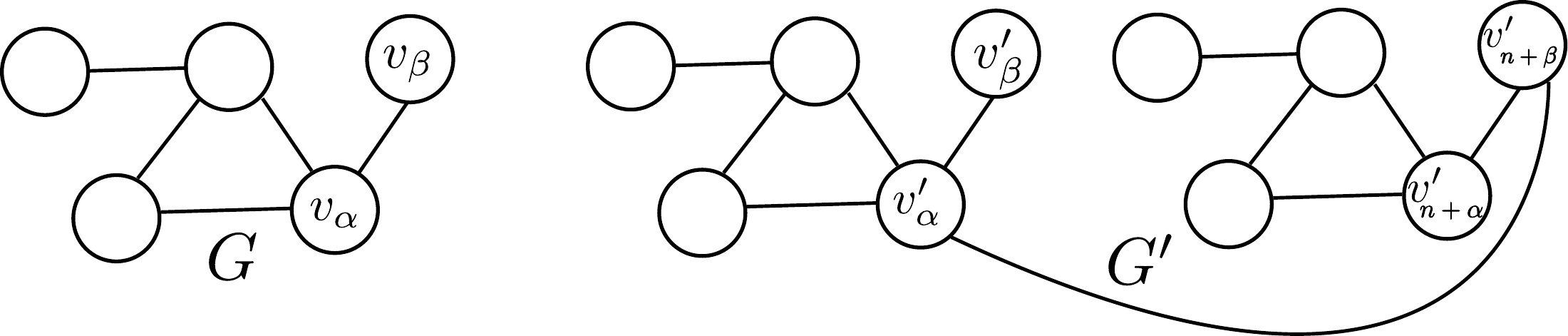}
\caption{An example of communication graphs $G$ and $G'$ ($n=5$)}
\label{fig:graph1}
\end{center}
\end{figure}

An example ($n=5$) of $G$ and $G'$ is shown in Figure \ref{fig:graph1}.

Consider a globally-fair execution $\Xi' = C'_0$, $C'_1$, $C'_2$, $C'_3$, $\ldots$ over $G'$ as follows:
\begin{itemize}
\item For $i \le t$, when $v_x$ interacts with $v_y$ at $C_i \rightarrow C_{i+1}$, $v'_x$ interacts with $v'_y$ at $C'_{2i} \rightarrow C'_{2i+1}$, and $v'_{x+n}$ interacts with $v'_{y+n}$ at $C'_{2i+1} \rightarrow C'_{2i+2}$.
\item After $C'_{2t}$, agents make interactions so that $\Xi'$ satisfies global fairness.
\end{itemize}

In the following, for configuration $C'$ in $\Xi'$ and configuration $C$ in $\Xi$, we say $C'$ of $V'_1$ (resp., $V'_2$) is equivalent to $C$ if $s(v'_x,C')=s(v_x,C)$ (resp., $s(v'_{x+n},C')=s(v_x,C)$) holds for any $v_x\in V$.
Observe that, by the definition of $\Xi'$, $C'_{2t}$ of $V'_1$ and $V'_2$ is equivalent to $C_t$.

From now on, by induction on the index of configuration, we prove the proposition that, for any configuration $C'_m$ that occurs after $C'_{2t}$, there is a configuration $C_a$ (resp., $C_b$) such that 1) $C'_m$ of $V'_1$ (resp., $V'_2$) is equivalent to $C_a$ (resp., $C_b$) and 2) $C_a$ (resp., $C_b$) appears infinitely often in $\Xi$.

The base case is $C'_{m} = C'_{2t}$.
Since $C_{t}$ appears infinitely often in $\Xi$ and $C'_{2t}$ of $V'_1$ and $V'_2$ is equivalent to $C_t$, the base case holds.

For the induction step, assume that, for $C'_m(m \ge 2t)$, there is a configuration $C_a$ (resp., $C_b$) that satisfies the conditions.
We consider two cases for an interaction at $C'_m \rightarrow C'_{m+1}$.
The first case considers an interaction of $v'_\alpha$ and $v'_{n+\beta}$, and the second case considers other interactions.

First, we consider the case that $v'_{\alpha}$ and $v'_{n+\beta}$ interact at $C'_m \rightarrow C'_{m+1}$.
By the assumption, we have $s(v'_\alpha,C'_m)=s(v_\alpha,C_a)=s_\alpha$, $s(v'_{n+\alpha},C'_m)=s(v_\alpha,C_b)=s_\alpha$, and $s(v'_{n+\beta}, C'_{m}) = s(v_{\beta}, C_{b})$.
Hence, since $v_{\alpha}$ does not change its state even when it interacts with $v_\beta$ at $C_b$, a transition rule $(s(v_{\beta}, C_{b}), s_\alpha) \rightarrow (s, s_\alpha)$ exists, where $s$ is some state.
This implies that, when $v'_{\alpha}$ and $v'_{\beta + n}$ interact at $C'_m$, $v'_{\alpha}$ keeps the state $s_\alpha$.
Thus, $C'_{m+1}$ of $V'_1$ is still equivalent to $C_a$.
Additionally, since $s(v'_{\alpha}, C'_{m}) = s(v'_{n+\alpha}, C'_{m}) = s(v_\alpha,C_b) = s_\alpha$ and $s(v'_{n+\beta},C'_m) = s(v_\beta,C_b)$ hold, $v'_{n+\beta}$ changes its state similarly to the case that $v_\beta$ interacts with $v_\alpha$ in $C_b$.
That is, letting $C_{b'}$ be a configuration immediately after $v_\alpha$ and $v_\beta$ interact at $C_b$, $C'_{m+1}$ of $V'_2$ is equivalent to $C_{b'}$.
Since $C_b$ occurs infinitely often in $\Xi$ and $\Xi$ is globally fair, $C_{b'}$ occurs infinitely often in $\Xi$.
Thus, the proposition holds for $C'_{m+1}$.

Next, we consider the case that at least one agent other than $v'_\alpha$ and $v'_{n+\beta}$ joins an interaction at $C'_m \rightarrow C'_{m+1}$.
By the definition, no edge other than $(v'_\alpha,v'_{n+\beta})$ connects $V'_1$ and $V'_2$.
Hence, if $v'_i$ and $v'_j$ interact at $C'_m \rightarrow C'_{m+1}$, either $v'_i \in V'_1 \land v'_j \in V'_1$ or $v'_i \in V'_2 \land v'_j \in V'_2$ holds.
In the former case, letting $C_{a'}$ be the configuration immediately after $v_i$ and $v_j$ interact at $C_a$, $C'_{m+1}$ of $V'_1$ is equivalent to $C_{a'}$ and $C'_{m+1}$ of $V'_2$ is still equivalent to $C_b$. 
In the latter case, letting $C_{b'}$ be the configuration immediately after $v_{i-n}$ and $v_{j-n}$ interact at $C_b$, $C'_{m+1}$ of $V'_1$ is still equivalent to $C_a$ and $C'_{m+1}$ of $V'_2$ is equivalent to $C_{b'}$. 
Since $C_a$ and $C_b$ occur infinitely often in $\Xi$ and $\Xi$ is globally fair, such $C_{a'}$ and $C_{b'}$ occurs infinitely often in $\Xi$ and thus the proposition holds in the case.

Since $\nred{V} - \nblue{V} = 1 $ holds after $C_h$ in $\Xi$, $\nred{V'_1} - \nblue{V'_1} = 1 $ and $\nred{V'_2} - \nblue{V'_2} = 1 $ holds after $C'_{2t}$ in $\Xi'$.
Thus, $\nred{V'} - \nblue{V'} = 2 $ holds after $C'_{2t}$ in $\Xi'$.
Since $\Xi'$ is globally fair, this is a contradiction.
\end{proof}

From now, by using Lemma \ref{lem:nocha}, we show the theorem.

\begin{theorem}
There exists no uniform bipartition protocol with three states and designated initial states over arbitrary communication graphs assuming global fairness.
\end{theorem}
\begin{proof}
For the purpose of contradiction, we assume that such a protocol $Alg$ exists.

Let $S=\{ s_1, s_2, s_3 \}$ be a state set of agents.
Without loss of generality, $f(s_1) = f(s_2) = red$ and $f(s_3) = blue$ hold.
Consider a globally-fair execution $\Xi$ of $Alg$ over graph $G$ such that the number of agents is odd and no more than $P/2$.
By Lemma \ref{lem:nocha}, after some stable configuration $C_t$ in $\Xi$, each agent changes its state infinitely often.
This implies that each agent with $s_3$ transitions to $s_1$ or $s_2$ after $C_t$. That is, each $blue$ agent transitions to $red$ state after $C_t$.
Since $C_t$ is stable, this is a contradiction.
\end{proof}

\subsection{Lower Bound for Symmetric Protocols under Global Fairness}

In this section, we show that, with arbitrary communication graphs, designated initial states, and no base station assuming global fairness, there exists no symmetric protocol with four states.
Recall that, with designated initial states and no base station, clearly any symmetric protocol never solves the problem if the number of agents $n$ is two.
Thus, we assume that $3 \le n \le P$ holds, where $P$ is a known upper bound of the number of agents.
Note that the symmetric protocol proposed in subsection \ref{sec:uppsym} solves the problem for $3 \le n \le P$.

In this subsection, we newly define $q \overset{sym}{\rightsquigarrow} q'$ as follow:

\begin{itemize}
\item For states $q$ and $q'$, we say $q \overset{sym}{\rightsquigarrow} q'$ if there exists a sequence of states $q = q_0, q_1, \cdots , q_k = q'$ such that, for any $i(0 \le i < k)$, 
transition rule $(q_i,q_i) \rightarrow (q_{i+1},q_{i+1})$ exists.
\end{itemize}

Moreover, we say two agents are homonyms if they have the same state.
Intuitively, $q \overset{sym}{\rightsquigarrow} q'$ means that an agent in state $q$ can transition to $q'$ by only interactions with homonyms.

\begin{theorem}
\label{thm:no4sym}
There exists no symmetric protocol for the uniform bipartition with four states and designated initial states over arbitrary graph assuming global fairness when $P$ is twelve or more.
\end{theorem}

For the purpose of contradiction, suppose that there exists such a protocol $Alg$. Let $R$ (resp., $B$) be a state set such that, for any $s \in R$ (resp., $s^{\prime} \in B$), $f(s) = red$ (resp., $f(s^{\prime}) = blue$) holds.
First, we show that the following lemma holds from Lemma~\ref{lem:nocha}.

\begin{lemma}
$|R|=|B|$ holds (i.e., $|R| = 2$ and $|B| = 2$ hold).
\end{lemma}
\begin{proof}
Assume, for the purpose of contradiction, that $|R| \neq |B|$ holds.
Without loss of generality, assume that $|R|=1$ and $|B|=3$ hold (clearly $R\neq \emptyset$ holds and thus only this combination is valid).
Moreover, let $ r $ (resp., $b_1$, $b_2$, and $b_3$) be a state belonging to $R$ (resp., $B$).

Consider a globally-fair execution $\Xi = C_0$, $C_1$, $C_2$, $\ldots$ of $Alg$ with a communication graph $G$ such that the number of agents $n$ is three.
From Lemma \ref{lem:nocha}, after some stable configuration $C_t$ in $\Xi$, each agent transitions its own state to another state infinitely often.
Hence, after $C_t$, some agent with $r$ transitions to some $blue$ state.
Since $C_t$ is stable, this is a contradiction.
\end{proof}

Let $ ini_r$ and $r $ (resp.,  $ini_b$ and $b$) be states belonging to $R$ (resp., $B$).
In addition, without loss of generality, assume that $ini_r$ is the initial state of agents.

From the property of symmetry, the following lemma holds.
\begin{lemma}
\label{lem:sym}
Consider a symmetric transition sequence $(ini_r,ini_r) \rightarrow (p_1,p_1)$,  $(p_1,p_1) \rightarrow (p_2,p_2)$, $(p_2,p_2) \rightarrow (p_3,p_3)$, $\ldots $ starting from $ini_r$.
For any $i$, $p_i \neq p_{i+1}$ holds.
\end{lemma}
\begin{proof}
For the purpose of contradiction, suppose that $p_i = p_{i+1}$ holds for some $i$.
Consider a complete communication graph $G=(V,E)$ such that the number of agents $n$ is four, where $V=\{ v_1$, $v_2$, $v_3$, $v_4 \}$.

Consider a globally-fair execution $\Xi$ as follows:
\begin{itemize}
\item $v_1$ (resp., $v_3$) interacts with $v_2$ (resp., $v_4$) $i$ times.
\item After that, make interactions so that $\Xi$ satisfies global fairness.
\end{itemize}

Since the initial state of agents is $ini_r$, all agents have $p_i$ after the $i$ interactions.
By the assumption, since $(p_i,p_i) \rightarrow (p_i,p_i)$ holds, every agent keeps state $p_i$ after that.
Hence, $\Xi$ cannot reach a stable configuration.
Since $\Xi$ is globally fair, this is a contradiction.
\end{proof}

Additionally, we can extend Lemma \ref{lem:sym} as follow:
\begin{lemma}
\label{lem:sym2}
There exists some $s_b\in B$ such that $(ini_r,ini_r) \rightarrow (s_b,s_b)$ and $(s_b,s_b) \rightarrow (ini_r,ini_r)$ hold.
\end{lemma}
\begin{proof}
Consider a globally-fair execution $\Xi = C_0$, $C_1$, $C_2$, $\ldots$ of $Alg$ with a complete communication graph $G$ such that the number of agents $n$ is six.
First, consider a state $q$ such that there are two or more agents with $q$ in a stable configuration $C_t$ of $\Xi$.
Observe that, for any state $q'$ such that $q \overset{sym}{\rightsquigarrow} q'$ holds, $f(q) = f(q')$ holds.
This is because, if such equation does not hold, it contradicts the definition of the stable configuration (i.e., it contradicts the fact that each agent cannot change its own color after a stable configuration).
Using this fact, we show that the lemma holds.

Since the number of states is four and Lemma \ref{lem:sym} holds, there are three possible symmetric transition sequences starting from $ini_r$ as follows:
\begin{enumerate}
\item For distinct states $ini_r$, $p_1$, $p_2$, and $p_3$, there exists a transition sequence $(ini_r,ini_r)$ $\rightarrow$ $(p_1,p_1)$, $(p_1,p_1) \rightarrow (p_2,p_2)$, $(p_2,p_2) \rightarrow (p_3,p_3)$, $(p_3,p_3) \rightarrow (x,x)$, $\dots$, where $x \in \{ ini_r, p_1, p_2 \}$.
\item For distinct states $ini_r$, $p_1$, and $p_2$, there exists a transition sequence $(ini_r$, $ini_r)$ $\rightarrow$ $(p_1,p_1)$, $(p_1,p_1) \rightarrow (p_2,p_2)$, $(p_2,p_2) \rightarrow (y,y)$, $\dots$, where $y \in \{ ini_r, p_1 \}$.
\item For distinct states $ini_r$ and $p_1$, there exists a transition sequence $(ini_r,ini_r)$ $\rightarrow$ $(p_1,p_1)$, $(p_1,p_1) \rightarrow (ini_r,ini_r)$, $(ini_r,ini_r) \rightarrow (p_1,p_1)$, $\dots$.
\end{enumerate}

\noindent\textbf{Case 1:}
Assume, for the purpose of contradiction, that the transition sequence 1 holds.
Let $p \in \{p_1, p_2, p_3\}$ be a state such that $f(p)= red$ holds. By the assumption, in a stable configuration $C_t$ of $\Xi$, $\# red$ is three.
This implies that there exist two agents with $ini_r$ or $p$.
If there exist two agents with $ini_r$ in $C_t$, they can transition to $p' \in \{p_1, p_2, p_3\}$ such that $f(p') = blue \neq f(ini_r)$ holds by only interactions with homonyms.
Hence, by the definition of stable configurations, there exists at most one agent with $ini_r$ in $C_t$ and thus there exist two or more agents with $p$.

We consider cases of $p=p_3$, $p=p_2$, and $p=p_1$.
In the case of $p=p_3$, $f(p_1)=f(p_2)=blue$ holds. By the transition sequence 1, $p$ can transition to $p_2$ by only interactions with homonyms.
By the definition of stable configurations, since $f(p) \neq f(p_2)$ holds, this case does not hold.
In the case of $p=p_2$, $f(p_1)=f(p_3)=blue$ holds. By the transition sequence 1, $p$ can transition to $p_3$ by only interactions with homonyms.
By the definition of stable configurations, since $f(p) \neq f(p_3)$ holds, this case does not hold.
In the case of $p=p_1$, $f(p_2)=f(p_3)=blue$ holds. By the transition sequence 1, $p$ can transition to $p_2$ by only interactions with homonyms.
By the definition of stable configurations, since $f(p) \neq f(p_2)$ holds, this case does not hold.
Thus, the transition sequence 1 does not hold.

\noindent\textbf{Case 2:} 
Assume, for the purpose of contradiction, that the transition sequence 2 holds.
We show that 1) $f(p_1)=f(p_2)=blue$ holds and 2) $y = p_1$ holds.
After that, from these facts, we show that 3) the transition sequence 2 does not hold.

First we show 1) $f(p_1)=f(p_2)=blue$ holds.
Assume, for the purpose of contradiction, that either $p_1$ or $p_2$ is a $red$ state and the other is a $blue$ state (since $ini_r$ is $red$ state, either $p_1$ or $p_2$ is $blue$ state).
In this case, since $ini_r$ can transition to $blue$ state by only interactions with homonyms, there exist at most one agent with $ini_r$ in $C_t$ of $\Xi$.
Hence, in $C_t$, there exists two $red$ agents with $p_1$ or $p_2$.
By the transition sequence 2, $p_1$(resp., $p_2$) can transition to $p_2$(resp., $p_1$) by only interactions with homonyms.
Since $f(p_1) \neq f(p_2)$ holds, this contradicts the definition of stable configuration.

Next, we show 2) $y = p_1$ holds.
For the purpose of contradiction, assume that $y=ini_r$ holds.
In this case, $p_1$ and $p_2$ can transition to $ini_r$ by only interactions with homonyms.
In addition, $f(ini_r)=red$ holds.
From 1) $f(p_1)=f(p_2)=blue$, since $\# blue$ is three in a stable configuration $C_t$ of $\Xi$, there exist two agents with $p_1$ or $p_2$.
These facts contradict the definition of stable configuration.

Finally, we show 3) the transition sequence 2 does not hold.

Consider a complete communication graph $\hat{G}=(\hat{V},\hat{E})$ such that the number of agents $n$ is six, where $\hat{V}=\{ \hat{v_0}$, $\hat{v_1}$, $\hat{v_2}$, $\ldots$, $\hat{v_{5}} \}$.
Moreover, consider a globally-fair execution $\hat{\Xi}$ as follows:
\begin{itemize}
\item $\hat{v_0}$, $\hat{v_1}$, and $\hat{v_2}$ interact with $\hat{v_3}$, $\hat{v_4}$, and $\hat{v_5}$ once, respectively.
\item After that, make interactions so that $\hat{\Xi}$ satisfies global fairness.
\end{itemize}

After the first item, all agents have $p_1$.
Hence, from 1) $f(p_1)=f(p_2)=blue$ and 2) $y = p_1$ (i.e., $(p_1,p_1) \rightarrow (p_2,p_2)$ and $(p_2,p_2) \rightarrow (p_1,p_1)$ hold), there exists transition rule $(p_1,p_2) \rightarrow (s_{r1},s_{r2})$ such that $f(s_{r1}) = f(s_{r2}) = blue$ does not hold.
Since transition rules $(p_1,p_1) \rightarrow (p_2,p_2)$ and $(p_2,p_2) \rightarrow (p_1,p_1)$ exist and only $p_1$ and $p_2$ are $blue$ states, a stable configuration $\hat{C_t}$ of $\hat{\Xi}$ can reach $\hat{C_{t'}}$ such that both $p_1$ and $p_2$ exist in $\hat{C_{t'}}$.
Additionally, since $(p_1,p_2) \rightarrow (s_{r1},s_{r2})$ exists, some agent can change its color after $\hat{C_{t'}}$.
Since $\hat{C_t}$ is stable, this is a contradiction.
Thus, the transition sequence 2 does not hold.

\noindent\textbf{Case 3:} 
Assume, for the purpose of contradiction, that $f(p_1)=red$ holds.
Consider a globally-fair execution $\Xi'$ of $Alg$ with a complete communication graph $G'=(V',E')$ such that the number of agents is six.
Since transition rules $(ini_r,ini_r) \rightarrow (p_1,p_1)$ and $(p_1,p_1) \rightarrow (ini_r,ini_r)$ exist and $f(ini_r)=f(p_1)=red$ holds, there exists transition rule $(ini_r,p_1) \rightarrow (s_{b1},s_{b2})$ such that $f(s_{b1}) = f(s_{b2}) = red$ does not hold.
Since transition rules $(ini_r,ini_r) \rightarrow (p_1,p_1)$ and $(p_1,p_1) \rightarrow (ini_r,ini_r)$ exist and only $ini_r$ and $p_1$ are $red$ state (and the number of $red$ agents is three in any stable configuration), a stable configuration $C'_t$ of $\Xi'$ can reach $C'_{t'}$ such that both $p_1$ and $ini_r$ exist in $C'_{t'}$.
In addition, since $(ini_r,p_1) \rightarrow (s_{b1},s_{b2})$ exists, some agent can change its color from $C'_{t'}$.
Since $C'_{t}$ is stable, this is a contradiction.
Therefore, $f(p_1) = blue$ holds and thus the lemma holds.
\end{proof}

Without loss of generality, assume that $(ini_r,ini_r) \rightarrow (ini_b,ini_b)$ and $(ini_b,ini_b) \rightarrow (ini_r,ini_r)$ exist.
For some population $V$, we denote the number of agents with $ini_r$ (resp., $ini_b$) belonging to $V$ as $\ninir{V}$ (resp., $\ninib{V}$).
Moreover, let $\nini{V}$ be the sum of $\ninir{V}$ and $\ninib{V}$.
When $V$ is clear from the context, we simply denote them as $\#ini_r$, $\#ini_b$, and $\#ini$, respectively.

From now on, we show that Theorem~\ref{thm:no4sym} holds if the following lemmas and corollary hold.
We show the detail of these intermediate proofs later.

\begin{lemma}
\label{lem:ini}
There does not exist a transition rule such that $\#ini$ increases after the transition. 
\end{lemma}

\begin{lemma}
\label{lem:ininum}
Consider a globally-fair execution $\Xi$ of $Alg$ with some complete communication graph $G$.
After some configuration in $\Xi$, $\#ini \le 1$ holds.
\end{lemma}

\begin{corollary}
\label{cor:Ini}
Consider a state set $Ini = \{ ini_r$, $ini_b \}$.
When $s_1\notin Ini$ or $s_2\notin Ini$ holds, if transition rule $(s_1,s_2)\rightarrow(s'_1,s'_2)$ exists then $f(s_1)=f(s'_1)$ and $f(s_2)=f(s'_2)$ hold.
\end{corollary}

Consider a globally-fair execution $\Xi = C_0$, $C_1$, $C_2$, $\ldots$ of $Alg$ with a ring communication graph $G=(V, E)$ such that the number of agents is three, where $V = \{v_0$, $v_1$, $v_2\}$.
In a stable configuration of $\Xi$, either $\nblue{V} - \nred{V} =1$ or $\nred{V} - \nblue{V} =1$ holds.

First, consider the case of $\nblue{V} - \nred{V} =1$.

By Lemma \ref{lem:nocha}, $red$ agents keep exchanging $r$ for $ini_r$ in $\Xi$.
Moreover, by Lemma \ref{lem:ininum}, there exists a stable configuration in $\Xi$ such that $\#ini \le 1$ holds.
From these facts, there exists a stable configuration $C_{t}$ of $\Xi$ such that there exists exactly one agent that has $ini_r$. 
Without loss of generality, we assume that the agent is $v_0$.

Consider the communication graph $G'=(V',E')$ that includes four copies of $G$. 
The details of $G'$ are as follows:
\begin{itemize}
\item Let $V'= \{ v'_0$, $v'_1$, $v'_2$, $v'_3$, $\ldots$, $v'_{11} \}$. 
Moreover, we define a partition of $V'$ as $V'_1 = \{ v'_0$, $v'_1$, $v'_2\}$, $V'_2 = \{ v'_{3}$, $v'_4$, $v'_5\}$, $V'_3 = \{ v'_6$, $v'_7$, $v'_8 \}$, and $V'_4 = \{ v'_9$, $v'_{10}$, $v'_{11} \}$.
Additionally, let $V'_{red} = \{v'_0$, $v'_{3}$, $v'_{6}$, $v'_{9} \}$ be a set of agents that will have state $ini_r$.
\item $E'=\{(v'_x,v'_y)$, $(v'_{x+3},v'_{y+3})$, $(v'_{x+6},v'_{y+6})$, $(v'_{x+9},v'_{y+9}) \in V' \times V' \mid (v_x, v_y) \in E \} \cup \{(v'_{x}, v'_{y}) \in V' \times V' \mid x$, $y \in \{0$, $3$, $6$, $9 \} \}$.
\end{itemize}

An image of $G$ and $G'$ is shown in Figure \ref{fig:graph2}. 

\begin{figure}[t!]
\begin{center}
\includegraphics[scale=0.35]{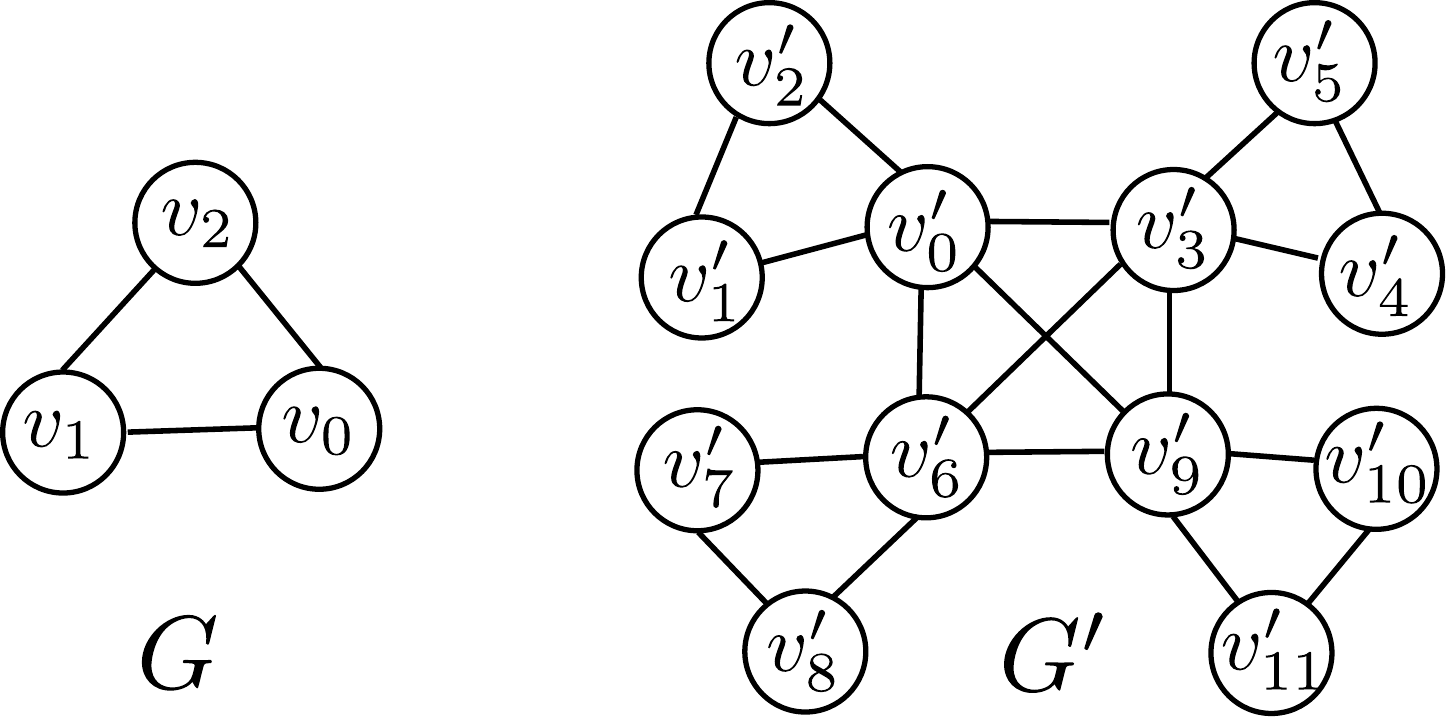}
\caption{An image of graphs $G$ and $G'$}
\label{fig:graph2}
\end{center}
\end{figure}

Consider the following execution $\Xi' = C'_0$, $C'_1$, $C'_2$, $\ldots$ of $Alg$ with $G'=(V',E')$.
\begin{itemize}
\item For $i \le t$, when $v_x$ interacts with $v_y$ at $C_i \rightarrow C_{i+1}$, $v'_x$ interacts with $v'_y$ at $C'_{4i} \rightarrow C'_{4i+1}$, $v'_{x+3}$ interacts with $v'_{y+3}$ at $C'_{4i+1} \rightarrow C'_{4i+2}$, $v'_{x+6}$ interacts with $v'_{y+6}$ at $C'_{4i+2} \rightarrow C'_{4i+3}$, and $v'_{x+9}$ interacts with $v'_{y+9}$ at $C'_{4i+3} \rightarrow C'_{4i+4}$.
\item After $C'_{4t}$, make interactions between agents in $V'_{red}$ until agents in $V'_{red}$ converge and $\nini{V'_{red}} \le 1$ holds.
We call the configuration $C'_{t'}$.
\item After $C'_{t'}$, make interactions so that $\Xi'$ satisfies global fairness.
\end{itemize}

Until $C'_{4t}$, agents in $V'_1$, $V'_2$, $V'_3$, and $V'_4$ behave similarly to agents in $V$ from $C_0$ to $C_{t}$. 
This implies that, in $C'_{4t}$, every agent in $V'_{red}$ has state $ini_r$.
From Lemma \ref{lem:ininum}, since $ini_r$ is the initial state of agents, it is possible to make interactions between agents in $V'_{red}$ until agents in $V'_{red}$ converge and $\nini{V'_{red}} \le 1$ holds.
Moreover, since $v_0$ is the only agent that has $ini_r$ in $C_{t}$, no agent in $V'_i \backslash V'_{red} (1 \le i \le 4)$ has state $ini_r$ or $ini_b$ in $C'_{4t}$.
Hence, $\#ini \le 1$ holds in $C'_{t'}$.
By Corollary \ref{cor:Ini}, if $\#ini \ge 2$ does not hold, no agent can change its color.
Thus, since $\#ini \le 1$ holds after $C'_{t'}$ by Lemma \ref{lem:ini}, no agent can change its color after $C'_{t'}$.
Since $v_1$ and $v_2$ are $blue$ in $C_t$, $v'_1$, $v'_2$, $v'_4$, $v'_5$, $v'_7$, $v'_8$, $v'_{10}$, and $v'_{11}$ are $blue$ in $C'_{t'}$.
In addition, $\nblue{V'_{red}} = \nred{V'_{red}}$ holds.
Hence, $\nblue{V'} - \nred{V'} = 8$ holds.
Since no agent can change its color after $C'_{t'}$ and $\Xi'$ is globally fair, this is a contradiction.

Next, consider the case of $\nred{V} - \nblue{V} =1$.
In this case, we can prove in the same way as the case of $\nblue{V} - \nred{V} =1$.
However, in the case, we focus on $ini_b$ instead of $ini_r$.
That is, we assume that agents in $V'_{red}$ (i.e., $v'_0$, $v'_3$, $v'_6$, and $v'_9$) have $ini_b$ in $C'_{4t}$.
From $C'_{4t}$, we make $v'_0$ (resp., $v'_6$) interact with $v'_3$ (resp., $v'_9$) once.
Then, by Lemma \ref{lem:sym2}, all of them transition to $ini_r$. 
After that, since all agents in $V'_{red}$ have $ini_r$, we can construct an execution such that only agents in $V'_{red}$ interact and eventually $\nini{V'_{red}} \le 1$ holds.
As a result, we can lead to contradiction in the same way as the case of $\nblue{V} - \nred{V} =1$.

\subsection*{Proofs of Lemmas \ref{lem:ini} and \ref{lem:ininum}, and Corollary \ref{cor:Ini}}

From now on, we prove Lemmas \ref{lem:ini} and \ref{lem:ininum}, and Corollary \ref{cor:Ini}.
First, from Lemmas \ref{lem:nocha} and \ref{lem:sym2}, we prove the following lemmas.

\begin{lemma}
\label{lem:sta}
Consider a globally-fair execution $\Xi$ of $Alg$ with some complete communication graph $G=(V, E)$.
In any stable configuration of $\Xi$, there is at most one agent with $ini_r$ and at most one agent with $ini_b$.
\end{lemma}
\begin{proof}
For the purpose of contradiction, assume that there exists a stable configuration $C_t$ of $\Xi$ such that there are more than one agent with $ini_r$ (or more than one agent with $ini_b$).
Since $G=(V, E)$ is a complete communication graph, two agents with $ini_r$ (or two agents with $ini_b$) can interact.
By Lemma \ref{lem:sym2}, $ini_r$(resp., $ini_b$) transitions to $ini_b$(resp., $ini_r$).
Since $f(ini_r) \neq f(ini_b)$ holds and $C_t$ is stable, this is a contradiction.
\end{proof}

\begin{lemma}
\label{lem:rb}
There exist transition rules $(r,r) \rightarrow (r,r)$, $(b,b) \rightarrow (b,b)$, and $(r,b) \rightarrow (r,b)$.
\end{lemma}
\begin{proof}
Consider a globally-fair execution $\Xi$ of $Alg$ with a complete communication graph $G=(V, E)$ such that the number of agents $n$ is seven.
Consider a stable configuration $C_{t_1}$ (resp., $C_{t_2}$) of $\Xi$ such that there exists exactly one agent with $ini_r$ (resp., $ini_b$) and $C_{t_1}$ (resp., $C_{t_2}$) occurs infinitely often in $\Xi$.
From Lemmas \ref{lem:sta} and \ref{lem:nocha}, $C_{t_1}$ and $C_{t_2}$ exist.
Note that, since the number of agents must be $P/2$ or less to use Lemma \ref{lem:nocha}, $P$ must be fourteen or more.
Since the number of agents $n$ is seven, there exist at least two agents with $r$ and at least two agents with $b$ in $C_{t_1}$ and $C_{t_2}$.
Hence, since $G$ is a complete graph and $C_{t_1}$ and $C_{t_2}$ are stable, if $(r,r) \rightarrow (r',r')$, $(b,b) \rightarrow (b',b')$, and $(r,b) \rightarrow (r'',b'')$ exist, $f(r) = f(r') = f(r'')$ and $f(b) = f(b')= f(b'')$ hold.
In addition, by Lemma \ref{lem:sta}, $r'$ and $r''$ are not $ini_r$, and, $b'$ and $b''$ are not $ini_b$.
Therefore, the lemma holds.
\end{proof}

From these lemmas, we can prove Lemma \ref{lem:ini}.

\setcounter{theorem}{21}

\begin{lemma}
There does not exist a transition rule such that $\#ini$ increases after the transition.
\end{lemma}
\begin{proof}
For the purpose of contradiction, assume that such a transition rule exists.
By Lemma \ref{lem:rb}, transition rules $(r, r) \rightarrow (r, r)$, $(b,b) \rightarrow (b,b)$, and $(r,b) \rightarrow (r,b)$ exist and these transition rules do not increase $\#ini$.
Hence, there exists $(rb, ini) \rightarrow (ini_1, ini_2)$ such that $rb \in \{r, b\}$, $ini \in \{ini_r, ini_b\}$, $ini_1 \in \{ini_r, ini_b\}$, and $ini_2 \in \{ini_r, ini_b\}$ hold. 

Consider a globally-fair execution $\Xi$ of $Alg$ with a complete graph $G=(V, E)$ such that the number of agents $n$ is five.
Moreover, consider a stable configuration $C_{t_1}$ (resp., $C_{t_2}$) such that there exists exactly one agent with $ini_r$ (resp., $ini_b$) and $C_{t_1}$ (resp., $C_{t_2}$) occurs infinitely often in $\Xi$.
From Lemmas \ref{lem:sta} and \ref{lem:nocha}, $C_{t_1}$ and $C_{t_2}$ exist.
Since the number of agents $n$ is five, there exists at least one agent with $r$ and at least one agent with $b$ in $C_{t_1}$ and $C_{t_2}$.
This implies that, an agent with $ini_r$ (resp., $ini_b$) can interact with an agent with $r$ in $C_{t_1}$ (resp., $C_{t_2}$).
Similarly, an agent with $ini_r$ (resp., $ini_b$) can interact with an agent with $b$ in $C_{t_1}$ (resp., $C_{t_2}$).
Moreover, since $G$ is a complete graph and $\Xi$ is globally fair, those interactions happen infinitely often.

First, we consider the case that $ini = ini_r$ and $rb = r$ hold.
Consider an interaction between an agent with $ini_r$ and an agent with $r$ in $C_{t_1}$.
Since any agent cannot change its color in a stable configuration, both agents transition to $ini_r$ by the interaction.
However, by Lemma \ref{lem:sta}, two agents cannot have $ini_r$ in any stable configuration. This is a contradiction.
Thus, $ini = ini_r$ and $rb = r$ do not hold.
In a similar way, $ini = ini_b$ and $rb = b$ do not hold. 

Next, we consider the case that $ini = ini_r$ and $rb = b$ hold.
Consider an interaction between an agent with $ini_r$ and an agent with $b$ in $C_{t_1}$.
Let $C_{t'_1}$ be configuration that can be obtained from $C_{t_1}$ by the interaction.
Since any agent cannot change its color after a stable configuration, one $ini_r$ and one $ini_b$ occur by the interaction.
Moreover, by Lemma \ref{lem:sta}, there exist exactly one $ini_r$ and exactly one $ini_b$ in $C_{t'_1}$. 
This implies that, since $C_{t'_1}$ is stable and $n$ is five, there exists at least one agent with $b$ and thus an agent with $ini_r$ can interact with an agent with $b$ in $C_{t'_1}$.
Hence, we can obtain $C_{t''_1}$ from $C_{t'_1}$ by making interaction between an agent with $ini_r$ and an agent with $b$.
Clearly, in $C_{t''_1}$, there exist two agents with $ini_b$. This contradicts Lemma \ref{lem:sta} and thus $ini = ini_r$ and $rb = b$ do not hold.
In a similar way, we can prove that $ini = ini_b$ and $rb = r$ do not hold.
Hence, for any $rb \in \{ r,b \}$ and $ini \in \{ ini_r,ini_b \}$, transition rule $(rb,ini) \rightarrow (ini_1,ini_2)$ with $ini_1$, $ini_2 \in \{ ini_r,ini_b \}$ does not exist. Consequently the lemma holds.
\end{proof}
\setcounter{theorem}{26}

By the existence of $(ini_r,ini_r) \rightarrow (ini_b,ini_b)$ and $(ini_b,ini_b) \rightarrow (ini_r,ini_r)$, we can prove the following lemma.
\begin{lemma}
\label{lem:inirb}
There exists a transition rule $(ini_r,ini_b) \rightarrow (x,y)$ such that $x \in \{ r, b \}$ or $y \in \{ r,b \}$ holds.
\end{lemma}
\begin{proof}
For the purpose of contradiction, assume that, if there is $(ini_r,ini_b) \rightarrow (x,y)$, $x \in \{ ini_r, ini_b \}$ and $y \in \{ ini_r, ini_b \}$ hold.

Consider a globally-fair execution $\Xi$ of $Alg$ with a complete communication graph $G=(V, E)$ such that the number of agents is three.
By Lemma \ref{lem:sta}, there exists at most one agent with $ini_r$ and at most one agent with $ini_b$ in a stable configuration of $\Xi$.
However, by the assumption and the existence of $(ini_r,ini_r) \rightarrow (ini_b,ini_b)$ and $(ini_b,ini_b) \rightarrow (ini_r,ini_r)$, all agents have $ini_r$ or $ini_b$ permanently in $\Xi$.
This is a contradiction.
\end{proof}

By Lemmas  \ref{lem:ini}, \ref{lem:sta}, and \ref{lem:inirb}, we show the proof of Lemma \ref{lem:ininum}.

\setcounter{theorem}{22}

\begin{lemma}
Consider a globally-fair execution $\Xi$ of $Alg$ with some complete communication graph $G$.
After some configuration in $\Xi$, $\#ini \le 1$ holds.
\end{lemma}
\begin{proof}
Consider $C_t$ such that $\#ini \le 2$ holds in $C_t$ and $C_t$ occurs infinitely often in $\Xi$.
By Lemma \ref{lem:sta}, such a configuration exists.
Moreover, by Lemma \ref{lem:ini}, $\#ini \le 2$ holds even after $C_t$.
First, consider the case that $\#ini \le 1$ holds in $C_{t}$.
By Lemma \ref{lem:ini}, $\#ini \le 1$ holds after $C_{t}$ and thus the lemma holds immediately in this case.
Next, consider the case that $\#ini = 2$ holds in $C_{t}$.
By Lemma \ref{lem:sta}, in $C_{t}$, there exists an agent $v_1$ (resp., $v_2$) with $ini_r$ (resp., $ini_b$).
By Lemma \ref{lem:inirb}, when $v_1$ interacts with $v_2$ at $C_{t} \rightarrow C_{t+1}$, $\#ini \le 1$ holds in $C_{t+1}$.
By global fairness, $C_{t+1}$ occurs infinitely often in $\Xi$.
By Lemma \ref{lem:ini}, $\#ini \le 1$ holds after $C_{t+1}$ and thus the lemma holds.
\end{proof}

\setcounter{theorem}{27}

From Lemmas \ref{lem:ininum} and \ref{lem:nocha}, we can obtain the following lemma.
\begin{lemma}
\label{lem:inirb2}
Let $ini$ and $rb$ be states such that $ini \in \{ ini_r, ini_b\}$ and $rb \in \{ r, b \}$ hold.
If $(ini, rb) \rightarrow (x,y)$ exists, $f(ini) = f(x)$ and $f(rb) = f(y)$ hold.
\end{lemma}
\begin{proof}
Assume, for the purpose of contradiction, that there exists $(ini, rb) \rightarrow (x,y)$ such that $f(ini) \neq f(x)$ or $f(rb) \neq f(y)$ holds.

Consider a globally-fair execution $\Xi$ of $Alg$ with a complete communication graph $G=(V, E)$ such that the number of agents $n$ is five.
By Lemmas \ref{lem:ininum} and \ref{lem:nocha}, there exists a stable configuration $C_{t_1}$ (resp., $C_{t_2}$) of $\Xi$ such that there is $v_r$ (resp., $v_b$) which is the only agent with $ini_r$ (resp., $ini_b$) in $C_{t_1}$ (resp., $C_{t_2}$) and there is $v_{rb}$ with $rb$.
By the assumption, when $v_r$ (resp., $v_b$) interacts with $v_{rb}$ at $C_{t_1} \rightarrow C_{t_1 + 1}$ (resp., $C_{t_2} \rightarrow C_{t_2 + 1}$), $v_r$ (resp., $v_b$) or $v_{rb}$ changes its color.
Since $C_{t_1}$ and $C_{t_2}$ are stable, this is a contradiction.
\end{proof}

From Lemmas \ref{lem:rb} and \ref{lem:inirb2}, we can obtain the following corollary.

\setcounter{theorem}{23}
\begin{corollary}
 Consider a state set $Ini = \{ ini_r$, $ini_b \}$.
When $s_1\notin Ini$ or $s_2\notin Ini$ holds, if transition rule $(s_1,s_2)\rightarrow(s'_1,s'_2)$ exists then $f(s_1)=f(s'_1)$ and $f(s_2)=f(s'_2)$ hold.
\end{corollary}
\setcounter{theorem}{28}

\subsection{Impossibility under Weak Fairness}

In this subsection, assuming arbitrary communication graphs and designated initial states and no base station, we show that there is no protocol that solves the problem under weak fairness.
Fischer and Jiang~\cite{fischer2006self} proved the impossibility of leader election for a ring communication graph.
We borrow their proof technique and apply it to the impossibility proof of a uniform bipartition problem. 


\begin{theorem}
There exists no protocol that solves the uniform bipartition problem with designated initial states and no base station under weak fairness assuming arbitrary communication graphs. 
\end{theorem}

\begin{proof}
For the purpose of contradiction, let us assume that there exists such a protocol $Alg$.

First, consider a ring $R_1$ with three agents $v_0$, $v_1$, and $v_2$.
Let $(v_0,v_1)$, $(v_1,v_2)$, and $(v_2,v_0)$ be the edges of $R_1$.
Furthermore, let $\Xi= C_0$, $C_1$, $C_2$, $\ldots$, $C_t$, $\ldots$ be an execution of $Alg$, where $C_t$ is a stable configuration.
Without loss of generality, we assume that $\# red = 1$ and $\# blue = 2$ hold in $C_t$.

Next, consider a ring $R_2$ with six agents such that two copies of $R_1$ are combined to form $R_2$.
Let $v'_0$, $v'_1$, $v'_2$, $v'_3$, $v'_4$, and $v'_5$ be agents of $R_2$, and let $(v'_0$, $v'_1)$, $(v'_1$, $v'_5)$, $(v'_5$, $v'_3)$, $(v'_3$, $v'_4)$, $(v'_4$, $v'_2)$, and $(v'_2$, $v'_0)$ be edges of $R_2$ (see Figure \ref{fig:ring}).

\begin{figure}[t!]
\begin{center}
\includegraphics[scale=0.23]{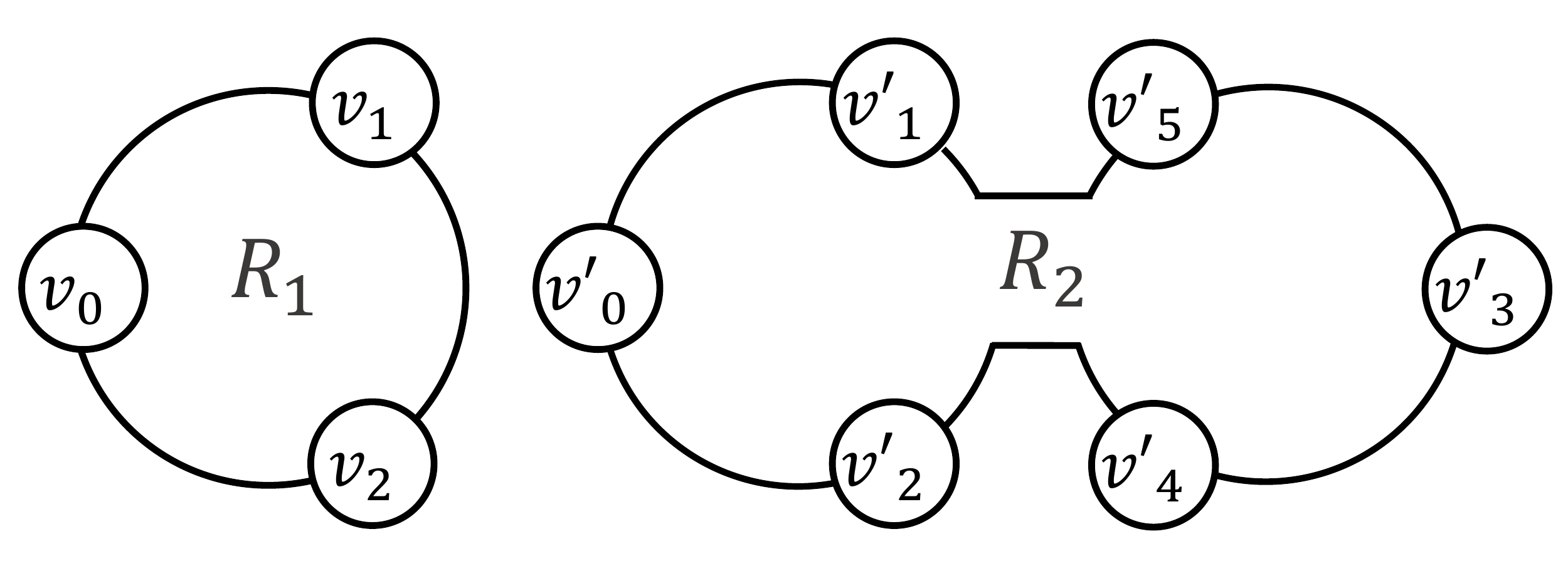}
\caption{Ring graphs $R_1$ and $R_2$}
\label{fig:ring}
\end{center}
\end{figure}

Now, let us construct the following execution $\Xi' =D_0$, $D'_0$, $D_1$, $D'_1$ $\ldots$.
\begin{itemize}
\item For $x$ and $y$ such that either $x = 0$ or $y=0$ holds, when $v_x$ interacts with $v_y$ at $C_i \rightarrow C_{i+1}$, $v'_x$ interacts with $v'_y$ at $D_i \rightarrow D'_{i}$, and $v'_{x+3}$ interacts with $v'_{y+3}$ at $D'_i \rightarrow D_{i+1}$.
\item When $v_1$ interacts with $v_2$ at $C_i \rightarrow C_{i+1}$, $v'_{1}$ interacts with $v'_5$ at $D_i \rightarrow D'_{i}$, and $v'_4$ interacts with $v'_2$ at $D'_i \rightarrow D_{i+1}$.
\end{itemize}

If configurations $C$ of $R_1$ and $D$ of $R_2$ satisfy the following condition, we say that those configurations are $equivalent$.
\begin{itemize}
\item For $i$ $(0 \le i \le 2)$, $s(v_i, C) = s(v'_i, D) = s(v'_{i+3}, D)$ holds.
\end{itemize}

From now on, by induction on the index of configuration, we show that $C_r$ and $D_r$ are equivalent for any $r \ge 0$.
Clearly $C_0$ and $D_0$ are equivalent, so the base case holds immediately.
For the induction step, we assume that $C_l$ and $D_l$ are equivalent, and then consider two cases of interaction at $C_l \rightarrow C_{l+1}$.

First we consider the case that, for $x$ and $y$ such that either $x = 0$ or $y=0$ holds, agents $v_x$ and $v_y$ interact at $C_l \rightarrow C_{l+1}$.
In this case, at $D_l \rightarrow D'_{l}$, $v'_x$ interacts with $v'_y$ and, at $D'_l \rightarrow D_{l+1}$, $v'_{x+3}$ interacts with $v'_{y+3}$. 
By the induction assumption, $s(v_x, C_l) = s(v'_x, D_l) = s(v'_{x+3}, D_l)$ and $s(v_y, C_l) = s(v'_y, D_l) = s(v'_{y+3}, D_l)$ hold.
Thus, agents $v'_{x}$ and $v'_{x+3}$ change their state similarly to $v_x$, and agents $v'_y$ and $v'_{y+3}$ change their state similarly to $v_y$. 
Hence, $C_{l+1}$ and $D_{l+1}$ are equivalent in this case.

Next, we consider the case that $v_1$ and $v_2$ interact at $C_l \rightarrow C_{l+1}$.
In this case, at $D_l \rightarrow D'_{l}$, $v'_{1}$ interacts with $v'_5$ and, at $D'_l \rightarrow D_{l+1}$, $v'_4$ interacts with $v'_2$. 
By the induction assumption, $s(v_1, C_l) = s(v'_1, D_l) = s(v'_4, D_l)$ and $s(v_2, C_l) = s(v'_2, D_l) = s(v'_5, D_l)$ hold.
Thus, agents $v'_1$ and $v'_4$ change their state similarly to $v_1$, and agents $v'_2$ and $v'_5$ change their state similarly to $v_2$. 
Hence, $C_{l+1}$ and $D_{l+1}$ are equivalent in this case.

Thus, $C_r$ and $D_r$ are equivalent for any $r \ge 0$.

This implies that, after $D_t$, $\# red = 2$ and $\# blue = 4$ hold. 
Moreover, since $\Xi$ is weakly fair, clearly $\Xi'$ is weakly fair.
This is a contradiction.
\end{proof}

\section{Concluding Remarks}

In this paper, we consider the uniform bipartition problem with designated initial states assuming arbitrary communication graphs.
We clarified the problem solvability, and even provided tight bounds (with respect to the number of states per agent) in the case of global fairness. 
Concretely, with no base station under global fairness, we proved that four is the minimum number of states per agent to enable asymmetric protocols, and five is the minimum number of states per agent to enable symmetric protocols.
With no base station under weak fairness, we proved the impossibility to obtain an asymmetric protocol. 
On the other hand, with a base station, we propose a symmetric protocol with $3P+1$ states under weak fairness, and a symmetric protocol with three states under global fairness. 

Our work raises interesting open problems: 
\begin{itemize}
\item Is there a relation between the uniform bipartition problem and other classical problems such as counting, leader election, and majority? We pointed out the reuse of some proof arguments, but the existence of a more systematic approach is intriguing.

\item What is the time complexity of the uniform bipartition problem?

\item Is the uniform bipartition problem in arbitrary communication graphs with arbitrary initial states feasible? In case the answer is yes, what is the space complexity?


\end{itemize}


\bibliography{ref}

\end{document}